\documentclass[11pt]{article}
\usepackage{amssymb, amsmath, amsthm}
\usepackage{natbib}
\usepackage{multirow}
\usepackage{times}
 \usepackage{subcaption}
\usepackage{graphicx}
\usepackage{algorithm}
\usepackage{algpseudocode}
\usepackage{setspace}
\usepackage{dsfont}
\usepackage{wasysym}
\usepackage{setspace}

\let\Algorithm\algorithm
\renewcommand\algorithm[1][]{\Algorithm[#1]\setstretch{0.8}}
\def\T{{\mathrm{\scriptstyle T}}}

\newcommand{\argmin}{\operatorname*{arg \ min}}

\theoremstyle{plain}

\newtheorem{theorem}{Theorem}
\newtheorem{assumption}{Assumption}
\newtheorem{lemma}{Lemma}
\newtheorem{corol}{Corollary}

\makeatletter
\def\ps@myheadings{%
    \let\@oddfoot\@empty\let\@evenfoot\@empty
    \def\@evenhead{\thepage\hfil\slshape\leftmark}%
    \def\@oddhead{{\slshape\rightmark}\hfil\thepage}%
    \let\@mkboth\@gobbletwo
    \let\sectionmark\@gobble
    \let\subsectionmark\@gobble
    }
  \if@titlepage
  \renewcommand\maketitle{\begin{titlepage}%
  \let\footnotesize\small
  \let\footnoterule\relax
  \let \footnote \thanks
  \null\vfil
  \vskip 60\p@
  \begin{center}%
    {\LARGE \@title \par}%
    \vskip 3em%
    {\large
     \lineskip .75em%
      \begin{tabular}[t]{c}%
        \@author
      \end{tabular}\par}%
      \vskip 1.5em%
    {\large \@date \par}
  \end{center}\par
  \@thanks
  \vfil\null
  \end{titlepage}%
  \setcounter{footnote}{0}%
}
\else
\renewcommand\maketitle{\par
  \begingroup
    \renewcommand\thefootnote{\@fnsymbol\c@footnote}%
    \def\@makefnmark{\rlap{\@textsuperscript{\normalfont\@thefnmark}}}%
    \long\def\@makefntext##1{\parindent 1em\noindent
            \hb@xt@1.8em{%
                \hss\@textsuperscript{\normalfont\@thefnmark}}##1}%
    \if@twocolumn
      \ifnum \col@number=\@ne
        \@maketitle
      \else
        \twocolumn[\@maketitle]%
      \fi
    \else
      \newpage
      \global\@topnum\z@   
      \@maketitle
    \fi
    \thispagestyle{plain}\@thanks
  \endgroup
  \setcounter{footnote}{0}%
}
\makeatother

\title{Shrinking characteristics of precision matrix estimators}
\author{Aaron J. Molstad$^1$ \hspace{2pt} and Adam J. Rothman$^2$\\
Biostatistics Program, Fred Hutchinson Cancer Research Center$^1$\\
  School of Statistics, University of Minnesota$^2$\\
  \texttt{amolstad@fredhutch.org}$^1$ \hspace{5pt} \texttt{arothman@umn.edu}$^2$}
\date{}
\begin{document}

\maketitle
\begin{abstract}
We propose a framework to shrink a user-specified characteristic of a precision matrix estimator that is needed to fit a predictive model. 
Estimators in our framework minimize
the Gaussian negative log-likelihood plus an
$L_1$ penalty on a linear or affine function evaluated at the optimization variable corresponding to the
precision matrix. We establish convergence rate 
bounds for these estimators and propose an alternating direction method of 
multipliers algorithm for their computation. 
Our simulation studies show that our estimators 
can perform better than competitors when they are used to fit predictive models.  In particular, we illustrate cases where 
our precision matrix estimators perform worse at estimating the population precision matrix but better at prediction.
\end{abstract}

\section{Introduction} 
The estimation of precision matrices is required to fit many statistical models. 
  Many papers written in the last decade have proposed shrinkage estimators of 
  the precision matrix when the number of variables $p$ is large. 
  \citet{pourahmadi2013high} and \citet{fan2016overview} provide 
  comprehensive reviews of large covariance and precision matrix estimation. 
  The main strategy used in many of these papers is to minimize the Gaussian negative 
  log-likelihood plus a penalty on the off-diagonal entries of the optimization 
  variable corresponding to the precision matrix. For example, \citet{yuan2007model} 
  proposed the $L_1$-penalized Gaussian likelihood precision matrix estimator defined by
  \begin{equation}\label{graphical_lasso} \argmin_{\Omega \in \mathbb{S}^p_+} 
  \left\{ {\rm tr}(S\Omega) - \log {\rm det}(\Omega) + \lambda \sum_{i \neq j}|\Omega_{ij}|\right\},
  \end{equation}
  where $S$ is the sample covariance matrix, $\lambda > 0$ is a tuning parameter, 
  $\mathbb{S}^p_+$ is the set of $p \times p$ symmetric and positive definite matrices, 
  $\Omega_{ij}$ is the $(i,j)$th entry of $\Omega$, 
  and ${\rm tr}$ and ${\rm det}$ denote the trace and determinant, respectively. 
  Other authors have replaced the $L_1$ penalty in \eqref{graphical_lasso} with the squared Frobenius norm 
  \citep{witten2009covariance,rothman2014existence} or non-convex
  penalties that also encourage zeros in the estimator
  \citep{lam2009sparsistency,fan2009network}. 

  To fit many predictive models, only a characteristic of the population precision matrix need be estimated.  For example, in binary linear discriminant analysis, the 
  population precision matrix is needed for prediction only through the product of the precision matrix and
   the difference 
  between the two conditional distribution mean vectors. Many authors have 
  proposed methods that directly estimate this characteristic
  \citep{cai2011direct,fan2012road, mai2012direct}.

We propose to estimate the precision matrix by shrinking the characteristic of the estimator that is needed for prediction. 
The characteristic we consider is a linear or affine function evaluated at 
  the precision matrix. The goal is to improve prediction performance. 
    Unlike methods that estimate the characteristic directly, our approach provides the practitioner with an estimate of the entire precision matrix, not just the characteristic. 
  In our simulation studies and data example, 
  we show that penalizing the characteristic needed for prediction can improve prediction 
  performance over competing sparse precision estimators like \eqref{graphical_lasso}, even 
  when the true precision matrix is very sparse. 
  In addition, estimators in our framework can be used in applications other than linear discriminant analysis. 

\vspace{-10pt}
  \section{Proposed method}\label{proposed}
  We propose to estimate the population precision matrix $\Omega_*$ with
  \begin{equation} \label{eq:estimator}
  \vspace{-3pt}
  \hat{\Omega} = \argmin_{\Omega \in \mathbb{S}_+^{p}}\left\{  {\rm tr}(S\Omega) 
  - \log {\rm det}(\Omega) +  \lambda |A \Omega B  - C |_1 \right\}, 
  \end{equation}
    \vspace{-3pt}
  where $A \in \mathbb{R}^{a \times p}$, $B \in \mathbb{R}^{p \times b}$, and $C \in \mathbb{R}^{a \times b}$  
  are user-specified matrices, and $|M|_1 = \sum_{i,j} |M_{ij}|$. Our estimator 
  exploits the assumption that $A \Omega_* B - C$ is sparse. 
 When $A$, $B$ and $C$ need to be estimated, we replace them in \eqref{eq:estimator} with their estimators. The matrix $C$ can serve as a shrinkage target for $A \hat{\Omega}B$, allowing practitioners to incorporate prior information.

\citet{dalal2014g} proposed a class of estimators similar to \eqref{eq:estimator} which replaces $A \Omega B - C$ with $T(\Omega)$, where $T$ is a symmetric linear transform. 

  Fitting the discriminant analysis model requires the estimation of one or more precision matrices. 
  In particular, the linear discriminant analysis model assumes 
  that the data are independent copies of the random pair $(X,Y)$, 
  where the support of $Y$ is $\left\{1, \dots, J\right\}$ and 
  \begin{equation}\label{Normal_Model_LDA} X \mid Y= j \sim {\rm N}_p
  \left(\mu_{*j}, \Omega_*^{-1}\right),\quad j=1, \dots, J,\end{equation}
  where $\mu_{*j} \in \mathbb{R}^p$ and $\Omega_*^{-1} \in \mathbb{S}^p_+$ are unknown. 
  To discriminate between response categories $l$ and $m$, only 
  the characteristic $\Omega_*(\mu_{*l} - \mu_{*m})$ is needed. 
  Methods that estimate this characteristic directly have been proposed \citep{cai2011direct,mai2012direct,fan2012road,mai2015multiclass}.
  These methods are useful in high dimensions because they perform variable selection. 
  For the $j$th variable to be non-informative for discriminating 
  between response categories $l$ and $m$, it must be that the $j$th 
    element of $\Omega_*(\mu_{*l} - \mu_{*m})$ is zero.  
  While these methods can perform well in classification and variable selection, they do not 
  actually fit the model in \eqref{Normal_Model_LDA}. 

  Methods for fitting \eqref{Normal_Model_LDA} specifically for 
  linear discriminant analysis either assume $\Omega_*$ is diagonal 
  \citep{bickel2004some} or that both $\mu_{*l} - \mu_{*m}$ and $\Omega_*$ are sparse \citep{guo2010simultaneous,xu2015covariance}. A method for fitting 
  \eqref{Normal_Model_LDA} and performing variable selection was proposed by 
  \citet{witten2009covariance}. They suggest a two-step procedure where one 
  first estimates $\Omega_*$, and then with the estimate $\bar{\Omega}$ fixed, 
  estimates each $\mu_{*j}$ by penalizing the characteristic $\bar{\Omega} \mu_j$, where $\mu_j$ is the optimization variable corresponding to $\mu_{*j}.$

  To apply our method to the linear discriminant analysis problem, we use \eqref{eq:estimator} with $A = I_p$, $C$ equal to the matrix of zeros, and $B$ 
  equal to the matrix whose columns are $\bar{x}_j - \bar{x}_{k}$ for all 
  $1 \leq j < k \leq J$, where $\bar{x}_j$ is the unpenalized maximum likelihood estimator of $\mu_{*j}$. For large values of the tuning parameter, this 
  would lead to an estimator of $\Omega_*$ such that $\hat{\Omega}(\bar{x}_j -\bar{x}_{k})$ is sparse. Thus our approach simultaneously fits \eqref{Normal_Model_LDA} and performs variable selection.

  Precision and covariance matrix estimators are also needed for 
  portfolio allocation. The optimal allocation based on the \citet{markowitz1952portfolio} 
  minimum-variance portfolio is proportional to $\Omega_*\mu_*$, where $\mu_*$ is the vector 
  of expected returns for $p$ assets and $\Omega_*$ is precision matrix for the returns. 
  In practice, one would estimate $\Omega_*$ and $\mu_*$  with their usual sample estimators 
  $\hat\Omega$ and $\hat{\mu}$. However, when $p$ is larger than the sample size, the usual sample estimator of 
  $\Omega_*$ does not exist, so regularization is necessary.  Moreover, \citet{brodie2009sparse} 
  argue that sparse portfolios, i.e., portfolios with fewer than $p$ active positions, are often desirable when $p$ is large. 
  While many have proposed 
  using sparse or shrinkage estimators of $\Omega_*$ or $\Omega_*^{-1}$ inserted in the 
  Markowitz criterion, e.g., \citet{xue2012positive}, this would not necessarily lead to 
  sparse estimators of $\Omega_* \mu_*$. 
  To achieve a sparse portfolio, \citet{chen2016regularized} proposed a method for estimating the characteristic 
  $\Omega_*\mu_*$ directly, but like the direct linear discriminant methods, 
  it does not lead to an estimate of $\Omega_*$. 
  For the sparse portfolio allocation application, we propose to estimate $\Omega_*$ using \eqref{eq:estimator} 
  with $A=I_p$, $C$ equal to the vector of zeros, and $B = \hat{\mu}$. 
    
  Another application is in linear regression where the $q$-variate response and $p$-variate predictors have a joint multivariate normal distribution. In this case, the regression coefficient matrix is $\Omega_* \Sigma_{*XY}$, where $\Omega_*$ is the marginal precision matrix for the predictors and $\Sigma_{*XY}$ is the cross-covariance matrix between predictors and response. We propose to estimate $\Omega_*$ using 
\eqref{eq:estimator} with $A = I_p$, $C$ equal to the matrix of zeros, and $B$ equal to 
the usual sample estimator of $\Sigma_{*XY}.$ Similar to the proposal of \citet{witten2009covariance}, this approach provides an alternative method for estimating regression coefficients using shrinkage estimators of the marginal precision matrix for the predictors. 

There are also applications where neither $A$ nor $B$ is equal to $I_p$. For example, in quadratic discriminant analysis, if one assumes that $\mu_{*j}^\T \Omega_{*j} \mu_{*j}$ is small, e.g., $\mu_{*j}$ is in the span of a set of eigenvectors of $\Omega_{*j}$ corresponding to small eigenvalues, then it may be appropriate to shrink entries of the estimates of $\Omega_{j*}$, $\Omega_{*j}\mu_{*j}$, and $\mu_{*j}^\T \Omega_{*j} \mu_{*j}$. For this application, we propose to estimate $\Omega_*$ using \eqref{eq:estimator} with $C$ equal to the matrix of zeros and $A^\T = B = (\bar{x}_j, \gamma I_p)$, for some tuning parameter $\gamma > 0$. 

  \section{Computation}\label{sec:computation}
  \subsection{Alternating direction method of multipliers algorithm}
  To solve the optimization in \eqref{eq:estimator}, we propose an 
  alternating direction method of multipliers algorithm with a modification 
  based on the majorize-minimize principle \citep{lange2016mm}. 
  Following the standard alternating direction method of multipliers approach 
  \citep{boyd2011distributed}, we rewrite \eqref{eq:estimator} as a constrained 
  optimization problem:
  \begin{equation} \label{eq: constrained_primal}
  \argmin_{\left(\Theta, \Omega\right) \in \mathbb{R}^{a \times b} \times \mathbb{S}^p_+}    \left\{ {\rm tr}(S\Omega) - \log{\rm det}(\Omega) +  \lambda |\Theta|_1 \right\} \quad \text{subject to } A\Omega B -  \Theta = C.  
  \end{equation}
  The augmented Lagrangian for \eqref{eq: constrained_primal} is defined by
  \begin{align*}
  \mathcal{F}_\rho(\Omega, \Theta, \Gamma) = {\rm tr}(S \Omega) 
  & - \log{\rm det}(\Omega) +  \lambda |\Theta|_1 \\
  & -  
  {\rm tr} \left\{ \Gamma^\T(A\Omega B - \Theta - C)\right\} + \frac{\rho}{2}\|A\Omega B - \Theta - C\|_F^2, 
  \end{align*} 
  where $\rho > 0$, $\Gamma \in \mathbb{R}^{a \times b}$ 
  is the Lagrangian dual variable, and $\|\cdot\|_F$ is the Frobenius norm.
 Let the subscript $k$ denote the $k$th iterate. 
  From \citet{boyd2011distributed}, to solve \eqref{eq: constrained_primal}, the alternating direction method of 
    multipliers algorithm uses the updating equations
  \begin{align}
  \Omega_{k+1} &= \argmin_{\Omega \in \mathbb{S}^p_+} 
  \mathcal{F}_{\rho}(\Omega, \Theta_{k}, \Gamma_{k}), \label{M_update_original} \\
  \Theta_{k+1} &= \argmin_{\Theta \in \mathbb{R}^{a \times b}}  
  \mathcal{F}_{\rho}(\Omega_{k+1}, \Theta, \Gamma_{k}), \label{soft_thresholding_step}\\
  \Gamma_{k+1} &= \Gamma_{k} -  \rho \left(A \Omega_{k+1}B - \Theta_{k+1} - C \right). \label{dual_variable_step}
  \end{align}
  The update in \eqref{M_update_original} requires its own iterative algorithm, which is complicated by the positive definiteness of the optimization variable.
   To avoid this computation, we replace \eqref{M_update_original} with an approximation based on the majorize-minimize principle. In particular, we replace $\mathcal{F}_\rho(\cdot, \Theta_k, \Gamma_k)$ with a majorizing function at the current iterate $\Omega_{k}$:
  \begin{equation}\Omega_{k+1} = \argmin_{\Omega \in \mathbb{S}^+_p} 
  \left\{\mathcal{F}_\rho(\Omega, \Theta_{k}, \Gamma_{k}) + \frac{\rho}{2}{\rm vec}(\Omega - \Omega_{k})^\T 
  Q {\rm vec}(\Omega - \Omega_{k})\right\},\label{majorizer}
  \end{equation}
  where
  $ Q =   \tau I - \left(A^\T A \otimes BB^\T\right),$
    $\tau$ is selected so that $Q \in \mathbb{S}^p_+$, $\otimes$ is the Kronecker product, and ${\rm vec}$ forms a vector by stacking the columns of its matrix argument. 
  Since $$ {\rm vec}(\Omega - \Omega_{k})^\T\left(A^\T A  \otimes BB^\T\right) 
  {\rm vec}(\Omega - \Omega_{k}) = {\rm tr}\left\{ A^\T A (\Omega - \Omega_{k}) BB^\T
   (\Omega - \Omega_{k})\right\},$$  we can rewrite \eqref{majorizer} as 
 $$
  \Omega_{k+1} = \argmin_{\Omega \in \mathbb{S}^+_p} \left[ \mathcal{F}_\rho(\Omega, \Theta_{k}, \Gamma_{k})  
  + \frac{\rho\tau}{2} \|\Omega - \Omega_{k}\|_F^2  - \frac{\rho}{2} {\rm tr}\left\{ A^\T A
   (\Omega - \Omega_{k}) BB^\T (\Omega - \Omega_{k}) \right\} \right],$$
  which is equivalent to
  \begin{equation}\label{eq:majorized_LINADMM}   
  \Omega_{k+1} = \argmin_{\Omega \in \mathbb{S}^+_p}   \left[ {\rm tr}\left\{ \left(S + G_{k} \right) \Omega\right\}
   - \log {\rm det}(\Omega) + \frac{\rho\tau}{2} \|\Omega - \Omega_{k}\|_F^2 \right], 
   \end{equation}
  where $G_{k} = \rho A^\T(A\Omega_{k} B-\rho^{-1}\Gamma_{k}- \Theta_{k} -  C)B^\T$. 
  The zero gradient equation for \eqref{eq:majorized_LINADMM} is 
  \begin{equation}\label{gradient_equation} S - \Omega_{k+1}^{-1} + 
  \frac{1}{2} \left( G_{k} +  G_{k}^\T \right) + \rho \tau \left(\Omega_{k+1} - \Omega_{k}\right)=0,
  \end{equation}
  whose solution is \citep{witten2009covariance,price2015ridge}
  $$\Omega_{k+1} = \frac{1}{2\rho\tau} U \left\{ -\Psi + \left(\Psi^2 + 4 \rho \tau I_p\right)^{1/2}\right\} U^\T,$$
  where $U \Psi U^\T$ is the eigendecomposition of $S + (G_{k} + G_{k}^\T)/2 - \rho \tau \Omega_{k}.$ 
  Our majorize-minimize approximation is a special case of the prox-linear
  alternating direction method of multiplier algorithm \citep{chen1994proximal,deng2016global}. Using the majorize-minimize approximation of \eqref{M_update_original} guarantees that $\mathcal{F}_\rho(\Omega_{k+1}, \Theta_k, \Gamma_k) \leq \mathcal{F}_\rho(\Omega_{k}, \Theta_k, \Gamma_k)$ and maintains the convergence properties of our algorithm \citep{deng2016global}.

  Finally, \eqref{soft_thresholding_step}  also has a closed form solution: 
  \begin{equation}\label{soft_update} 
  \Theta_{k+1} = {\rm soft}\left( A\Omega_{k+1}B- \rho^{-1}\Gamma_{k} - C,\rho^{-1}\lambda \right) \notag,
  \end{equation}
  where ${\rm soft}(x, \phi) = \max \left(|x|- \phi, 0\right) {\rm sign}(x)$. To summarize, we solve \eqref{eq:estimator} with the following algorithm. 


 \begin{algorithm}
  \caption{Alternating direction method of multipliers algorithm for \eqref{eq:estimator} } \label{ADMM_CharShrink}
  \noindent Initialize $\Omega_{(0)} \in \mathbb{S}_+^p$, $\Theta_{(0)}\in \mathbb{R}^{a \times b}$, 
  $\rho > 0$, and $\tau$ such that $Q$ is positive definite. Set $k=0.$ 
  Repeat Step 1 - 6 until convergence:
  \begin{tabbing}
  \textit{Step 1.} Compute $G_{k}= \rho A^\T( A\Omega_{k} B-\rho^{-1}\Gamma_k - \Theta_{k} - C)B^\T$; \\
  \textit{Step 2.} Decompose $S + 2^{-1} (G_{k} + G_{k}^\T) - \rho \tau\Omega_{k} = U \Psi U^\T$ 
  where $U$ is orthogonal and $\Psi$ is diagonal; \\
  \textit{Step 3.} Set $\Omega_{k+1} = (2\rho\tau)^{-1} U \left\{ -\Psi + (\Psi^2 + 4 \rho \tau I_p)^{1/2}\right\} U^\T$;\\
  \textit{Step 4.} Set $\Theta_{k+1} =  {\rm soft}(A\Omega_{k+1}B- \rho^{-1}\Gamma_{k} - C, \rho^{-1}\lambda )$;\\
  \textit{Step 5.} Set $\Gamma_{k+1} = \Gamma_{k} - \rho\left(A\Omega_{k+1} B - \Theta_{k+1} - C\right)$;\\
 \textit{Step 6.} Replace $k$ with  $k+1.$
\end{tabbing}
  \end{algorithm}

  \subsection{Convergence and implementation}\label{subsec:convergence}
  Using the same proof technique as in \citet{deng2016global}, one can show that 
  the iterates from Algorithm~\ref{ADMM_CharShrink} converge to their optimal values when a solution 
  to \eqref{eq: constrained_primal} exists. 

  In our implementation, we set $\tau = \varphi_{1}(A^\T A)\varphi_{1}(BB^\T) + 10^{-8}$, where $\varphi_{1}(\cdot)$ denotes the largest eigenvalue of its argument. 
  This computation is only needed once at the initialization of our algorithm. 
  We expect that in practice, the computational complexity of our algorithm will be 
  dominated by the eigendecomposition in Step 2, which requires $O(p^3)$ flops.

  To select the 
  tuning parameter to use in practice, we recommend using some type of cross-validation 
  procedure. For example, in the linear discriminant analysis 
  case, one could select the tuning parameter that minimizes the validation misclassification 
  rate or maximizes a validation likelihood.

  \section{Statistical Properties}\label{sec:stat_theory}
  We now show that by using the penalty in \eqref{eq:estimator}, we can estimate 
  $\Omega_*$ and $A \Omega_* B$ consistently in the Frobenius and $L_1$ norms, respectively. 
We focus on the case where $C$ is the $a \times b$ matrix of zeros. A nonzero $C$ substantially complicates the theoretical analysis and is left as a direction for future research. 

  Our results rely on assuming that  $A \Omega_* B$ is sparse. Define the set $\mathcal{G}$ as 
  the indices of $A\Omega_* B$ that are nonzero, i.e., 
  $$ \mathcal{G} = \left\{(i,j)\in \left\{1, \dots, a\right\} \times \left\{1 ,\dots, b\right\}:
   \left[ A \Omega_* B\right]_{ij} \neq 0 \right\}.$$ 
   Let the notation $[A \Omega_* B]_\mathcal{G} \in \mathbb{R}^{a \times b}$ denote the matrix 
   whose $(i,j)$th entry is equal to the $(i,j)$th of $A \Omega_* B$ if $(i,j) \in \mathcal{G}$ 
   and is equal to zero if $(i,j) \notin \mathcal{G}$. We generalize our results to the case that 
   $A$ and $B$ are unknown, and we use plug-in estimators of them in \eqref{eq:estimator}. 

  We first establish convergence rate bounds for known $A$ and $B$.
  Let $\sigma_{k}(\cdot)$ and $\varphi_{k}(\cdot)$ denote the $k$th largest singular value and eigenvalue of their arguments respectively. 
  Suppose that the sample covariance matrix used in \eqref{eq:estimator} is
  $S_n = n^{-1} \sum_{i=1}^n X_i X_i^\T,$ where $X_1, \dots, X_n$ are 
  independent and identically distributed $p_n$-dimensional random vectors with mean zero and covariance matrix $\Omega_*^{-1}$. 
  We will make the following assumptions:

  \begin{assumption}\label{A1s} For all $n$, there exists a constant $k_1$ 
  such that $$0 < k_1^{-1} \leq \varphi_{p_n}(\Omega_*) \leq \varphi_{1}(\Omega_*) \leq k_1 < \infty.$$ 
  \end{assumption} 

 \begin{assumption}\label{A2} For all $n$, there exists a constant $k_2$ 
  such that $\min \left\{ \sigma_{p_n}(A), \sigma_{p_n}(B) \right\} \geq k_2 > 0$. 
  \end{assumption}

  \begin{assumption}\label{A3} For all $n$, there exist positive constants $k_3$ and $k_4$ 
  such that $$\max_{j \in \left\{ 1, \dots, p_n \right\} } 
  E\left\{ {\rm exp}(tX_{1j}^2)\right\} \leq k_3 < \infty, \quad t \in  (-k_4, k_4).$$ 
  \end{assumption}
  Assumptions \ref{A1s} and \ref{A3} are common in the regularized precision matrix estimation literature; 
  Assumption \ref{A1s} was made by \citet{bickel2008regularized}, \citet{rothman2008sparse} and \citet{lam2009sparsistency},  and Assumption \ref{A3} holds if $X_1$ is multivariate normal. 
  Assumption \ref{A2} requires that $A$ and $B$ are both rank $p_n$, 
  which has the effect of shrinking every entry of $\hat{\Omega}$. 
  The convergence rate bounds we establish also depend on the quantity
  $$\xi(p_n, \mathcal{G}) = \sup_{M \in \mathbb{S}^{p_n}, M \neq 0 } 
  \frac{| \left[A M B\right]_{\mathcal{G}} |_1}{\|M\|_F},$$
  where $\mathbb{S}^{p_n}$ is the set of symmetric $p_n \times p_n$ matrices.  
  \citet{negahban2009unified} defined a similar and more general 
  quantity and called it a compatibility constant. 
  \begin{theorem}\label{knownAB_consist}
  Under Assumptions \ref{A1s}--\ref{A3}, if 
  $\lambda_n = K_1 (n^{-1}\log p_n)^{1/2}$, $K_1$ is sufficiently large, and 
  $\xi^2(p_n, \mathcal{G}) \log p_n = o(n)$, then (i)
  $\| \hat{\Omega} - \Omega_* \|_F =  O_P \{ \xi(p_n, \mathcal{G}) (\log p_n/n)^{1/2}\} $
  and
   (ii) $|A\hat{\Omega}B - A\Omega_*B |_1 = O_P\{  \xi^2(p_n, \mathcal{G}) ( \log p_n/n)^{1/2} \}$.
  \end{theorem}
  The quantity $\xi(p_n, \mathcal{G})$ can be used to recover known results for 
  special cases of \eqref{eq:estimator}. For example, when $A$ and $B$ are identity matrices, 
  $\xi(p_n,\mathcal{G}) = {s_n}^{1/2}$, where $s_n$ is the number of 
  nonzero entries in $\Omega_*$. This special case was established by 
  \citet{rothman2008sparse}. We can simplify the results of Theorem 
  \ref{knownAB_consist} for case that $A \Omega_* B$ has $g_n$ nonzero entries 
  by introducing an additional assumption:
  \begin{assumption}\label{A4}For all $n$, there exists a constant $k_5$ such that 
  $$\sup_{M \in \mathbb{S}^{p_n}, M \neq 0} 
  \frac{ \| [A M B]_\mathcal{G} \|_F}{\|M\|_F} \leq k_5 < \infty.$$ 
  \end{assumption}
  Assumption \ref{A4} is not the same as bounding $\xi(p_n,\mathcal{G}),$ 
  because the numerator uses the Frobenius norm instead of the $L_1$ norm. 
  This requires that for those entries of $A \Omega_* B$ which are nonzero, 
  the corresponding rows and columns of $A$ and $B$, do not have magnitudes too large as $p_n$ grows. 
  \begin{corol}\label{remark_gn}
  Under the conditions of Theorem \ref{knownAB_consist} and Assumption \ref{A4}, and if $A \Omega_* B$ has $g_n$ nonzero entries, then (i)
  $ \| \hat{\Omega} - \Omega_* \|_F = O_P\{ (g_n \log p_n/n)^{1/2}\}$ 
  and \\
  (ii) $| A\hat{\Omega}B - A\Omega_*B|_1 = O_P\{ (g_n^2 \log p_n/n)^{1/2}\}.$
  \end{corol}

  In practice, $A$ and $B$ are often unknown and must be estimated by $\hat{A}_n$ and $\hat{B}_n$, say. 
  In this case, we estimate $A\Omega_*B$ with $\hat{A}_n \tilde{\Omega} \hat{B}_n$, where 
  \begin{equation}\label{estimator_AB_unknown} 
  \tilde{\Omega} = \argmin_{\Omega \in \mathbb{S}_+^{p_n}}
  \left\{  {\rm tr}(S_n\Omega) - \log {\rm det}(\Omega) +  \lambda_n |\hat{A}_n \Omega \hat{B}_n|_1 \right\}.
  \end{equation}


  To establish convergence rate bounds for $\hat{A}_n \tilde{\Omega} \hat{B}_n$, we require an assumption on the asymptotic properties of $\hat{A}_n$ and $\hat{B}_n$:
   \begin{assumption}\label{A5} There exist sequences $a_n = o(1)$ and $b_n = o(1)$ such that
  $ |(A - \hat{A}_n)A^{+} |_1 = O_P\left(a_n\right)$ and $|B^{+}(B - \hat{B}_n)|_1  = O_P\left(b_n\right)$,
  where $A^{+}$ and $B^{+}$ are the Moore--Penrose pseudoinverses of $A$ and $B$. 
  \end{assumption}
 The convergence rate bounds we establish for \eqref{estimator_AB_unknown} also depend on the quantity $\Phi_n = \max (a_n, b_n) |\left[A \Omega_* B\right]_\mathcal{G}|_1,$ which partly controls the additional error incurred by using estimates of $A$ and $B$ in \eqref{estimator_AB_unknown}.

  \begin{theorem}\label{unknownAB_consist}
  Under Assumptions \ref{A1s}--\ref{A3} and \ref{A5}, if $\lambda_n = K_2 (\log p_n/n)^{1/2}$, $K_2$ is sufficiently large,
  $\xi^2(p_n, \mathcal{G})\log p_n = o(n)$, 
  and 
  $\Phi_n^2 \log p_n = o(n)$, 
  then \\
  (i) $\| \tilde{\Omega} - \Omega\|_F  =  O_P \{ \xi(p_n, \mathcal{G})(\log p_n/n)^{1/2} +   
  \Phi_n^{1/2} (\log p_n/n)^{1/4}\}$ and \\
  (ii) $|\hat{A}_n \tilde{\Omega} \hat{B}_n - A\Omega_* B|_1 = O_P\{ \xi^2(p_n, \mathcal{G}) (\log p_n/n)^{1/2} + \Phi_n^{1/2}  \xi(p_n, \mathcal{G}) (\log p_n/n)^{1/4} + \Phi_n\}.$
  \end{theorem}
  The convergence rate bounds in Theorem \ref{unknownAB_consist} are the sums of the statistical errors from Theorem \ref{knownAB_consist}
  plus additional errors from estimating $A$ and $B$.


  \section{Simulation studies}
    
  \subsection{Models}
  We compare our precision matrix estimator to competing estimators when they are used to fit the linear discriminant analysis model. For 100 independent replications, we generated a realization of $n$
  independent copies of $(X,Y)$ defined in \eqref{Normal_Model_LDA}, where $\mu_{*j} =  \Omega_*^{-1} \beta_{*j}$ and ${\rm pr}(Y = j) = 1/J$ for $j = 1, \dots, J$. Using this construction, because $\beta_{*l} - \beta_{*m} = \Omega_*(\mu_{*l} - \mu_{*m})$, if the $k$th element of $\beta_{*l} - \beta_{*m}$ is zero, then the $k$th variable is non-informative for discriminating between response categories $l$ and $m$.
    
    For each $J \in \left\{3, \dots, 10\right\}$, we partition our $n$ observations into a training set of size $25 J$, a validation set of size $200$, and a test set of size 1000. We considered two models for $\Omega_*^{-1}$ and $\beta_{*j}.$ 
  Let $\mathds{1}(\cdot)$ be the indicator function. \\

  \textit{Model 1.} We set $\beta_{*j,k} = 1.5 \hspace{2pt}\mathds{1} \left[ k \in \left\{4(j-1) + 1, \dots, 4j\right\}\right],$
  so that for any pair of response categories, only eight 
  variables were informative for discrimination. 
  We set $\Omega^{-1}_{*a,b} = .9^{|a-b|}$, so that $\Omega_*$ was tridiagonal. \\

  \textit{Model 2.}  We set $\beta_{*j,k} = 2 \hspace{2pt}\mathds{1} \left[k \in \left\{5(j-1) + 1, \dots, 5j\right\} \right],$
  so that for any pair of response categories, only ten variables were informative for discrimination. 
  We set $\Omega_*^{-1}$ to be block diagonal: the block corresponding 
  to the informative variables, i.e., the first $5J$ variables, had 
  off-diagonal entries equal to 0.5 and diagonal entries equal to one. The block submatrix corresponding 
  to the $p - 5J$ non-informative variables had $(a,b)$th entry equal to $0.5^{|a-b|}.$\\

  For both models, sparse estimators of $\Omega_*$ should perform well because the population precision 
  matrices are very sparse and invertible.
  The total number of informative variables is $4J$ and $5J$ in Models 1 and 2 respectively, 
  so a method like that proposed by \citet{mai2015multiclass}, which selects 
  variables that are informative for all pairwise response category comparisons, may perform poorly when $J$ is large. 

     \begin{figure}[t!]
   \centerline{\hfill\makebox[2.5in]{(a) Model 1}
      \hfill\makebox[2.5in]{(b) Model 2}\hfill}
  \centerline{\hfill
      \includegraphics[width=2.5in]{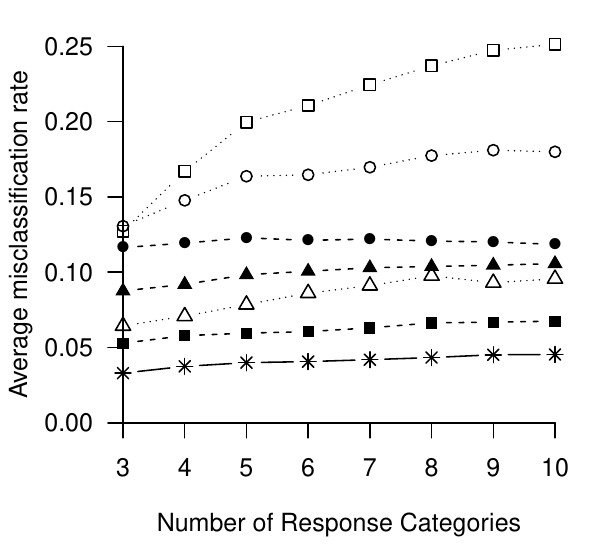}
      \hfill
     \includegraphics[width=2.5in]{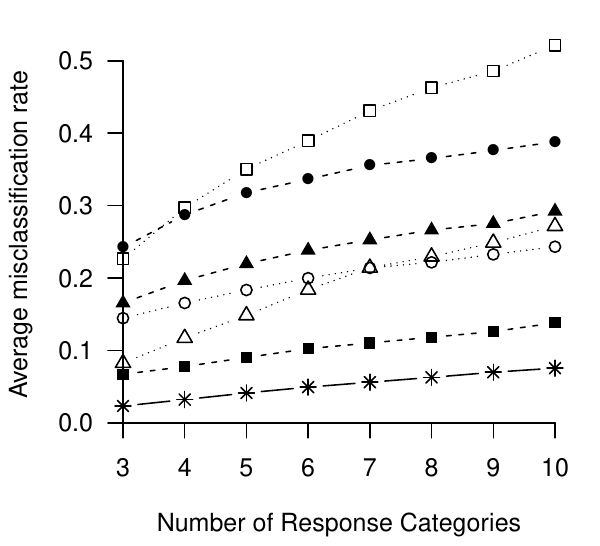}
      \hfill}  
      \centerline{\hfill\makebox[2.5in]{(c) Model 1}
      \hfill\makebox[2.5in]{(d) Model 2}\hfill}     
  \centerline{\hfill
      \includegraphics[width=2.5in]{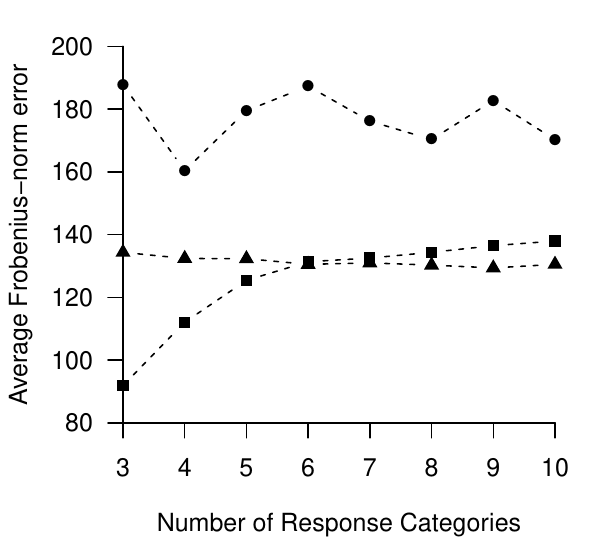}
      \hfill
      \includegraphics[width=2.5in]{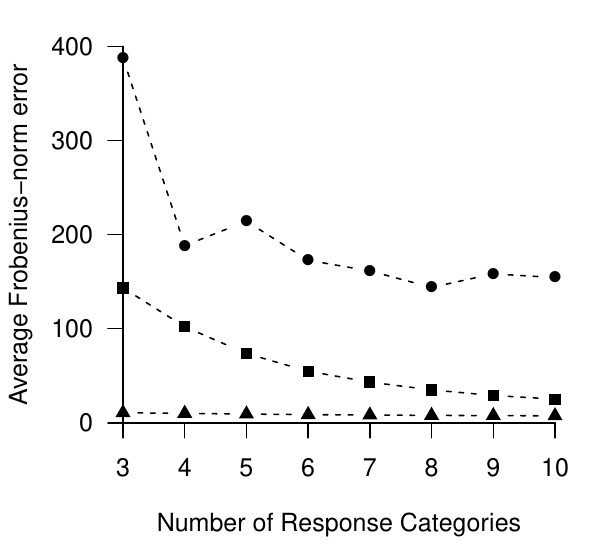}
      \hfill}           
  \caption{Misclassification rates and Frobenius norm error averaged over 100 replications with $p=200$ for Models 1 and 2. The methods displayed are the estimator we proposed in Section \ref{proposed} (dashed and $\blacksquare$), the $L_1$-penalized Gaussian likelihood estimator (dashed and $\blacktriangle$), the Ledoit-Wolf-type estimator $\hat{\Omega}^{-1}_{\rm LW}$ (dashed and $\newmoon$), Bayes (solid and {\large $\ast$}), the method proposed by \citet{guo2010simultaneous} (dots and $\fullmoon$), the method proposed by \citet{mai2015multiclass} (dots and $\triangle$), and the method proposed by \citet{witten2011penalized} (dots and $\square$).}
   \label{fig:misclass_mod1_2_later}
  \end{figure}

  \subsection{Methods}\label{sec:methods}
    We compared several methods in terms of classification accuracy on the test set. 
  We fit \eqref{Normal_Model_LDA} using the the
  sparse na\"ive Bayes estimator proposed by \citet{guo2010simultaneous} with 
  tuning parameter chosen to minimize misclassification rate on the validation set; the Bayes rule, i.e., $\Omega_*$, 
  $\mu_{*j}$, and ${\rm pr}(Y = j)$ known for $j = 1, \dots, J$. 
  We also fit \eqref{Normal_Model_LDA} using the ordinary sample means and using 
  the precision matrix estimator proposed in Section \ref{proposed} 
    with tuning parameter chosen 
  to minimize misclassification rate on the validation set and $B$ estimated using the sample means; 
  the $L_1$-penalized Gaussian likelihood precision matrix estimator 
  \citep{yuan2007model,rothman2008sparse,friedman2008sparse} with 
  tuning parameter chosen to minimize the misclassification rate of the 
  validation set; and a covariance matrix estimator similar to that proposed by \citet{ledoit2004well}, which is defined by
  $\hat{\Omega}^{-1}_{\rm LW} = \alpha S + \gamma (1 - \alpha) I_p,$
  where $(\alpha, \gamma) \in (0,1) \times (0, \infty)$ were chosen to 
  minimize the misclassification rate of the validation set. 
  The $L_1$-penalized Gaussian likelihood precision matrix estimator we used penalized the diagonals. 
  With our data-generating models, we found this performed better at classification than \eqref{graphical_lasso}, which does not penalize the diagonals. 
  We also tried two Fisher-linear-discriminant-based methods applicable to 
  multi-category linear discriminant analysis: the sparse linear discriminant 
  method proposed by \citet{witten2011penalized} with tuning parameter and 
  dimension chosen to minimize the misclassification rate of the validation set; and 
  the multi-category sparse linear discriminant method proposed by \citet{mai2015multiclass} 
  with tuning parameter chosen to minimize the misclassification rate of the validation set.

   \begin{figure}[t!]
   \centerline{\hfill\makebox[2.5in]{(a) Model 1}
      \hfill\makebox[2.5in]{(b) Model 2}\hfill}
  \centerline{\hfill
      \includegraphics[width=2.5in]{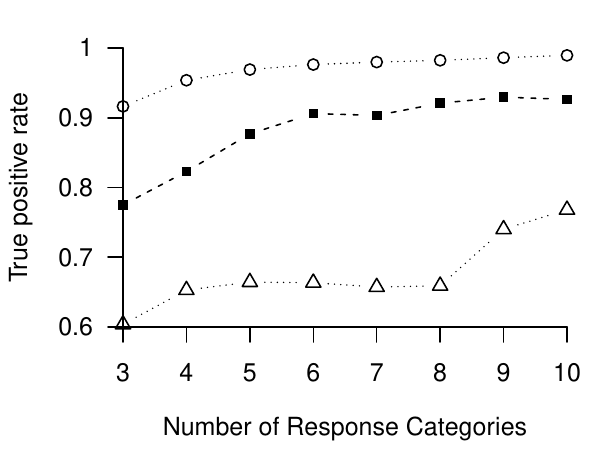}
      \hfill
     \includegraphics[width=2.5in]{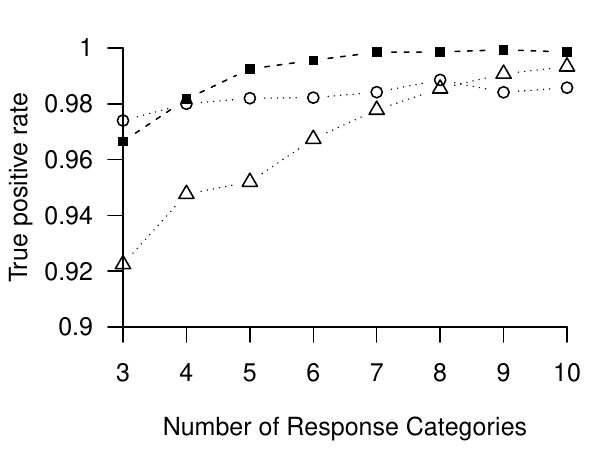}
      \hfill}  
      \centerline{\hfill\makebox[2.5in]{(c) Model 1}
      \hfill\makebox[2.5in]{(d) Model 2}\hfill}     
  \centerline{\hfill
      \includegraphics[width=2.5in]{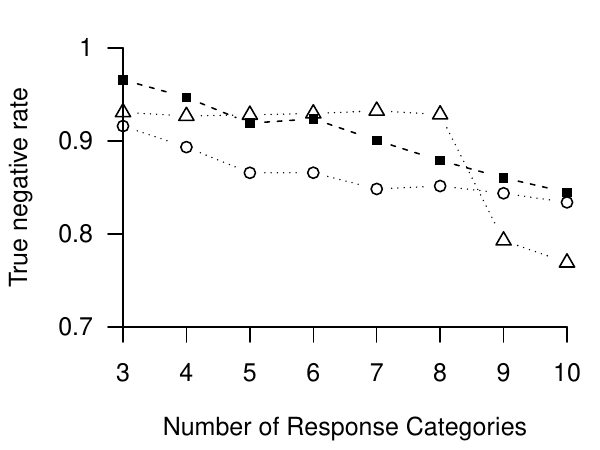}
      \hfill
      \includegraphics[width=2.5in]{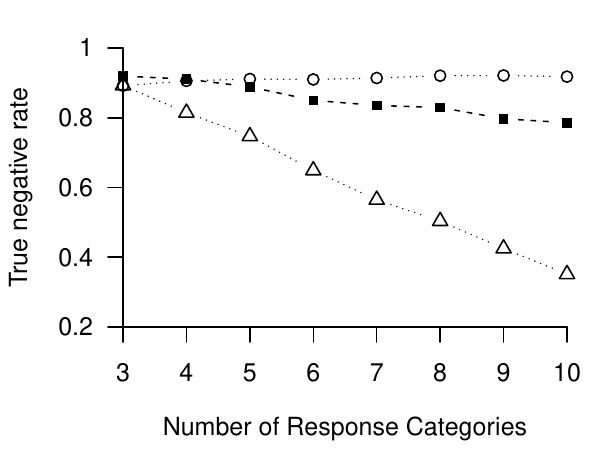}
      \hfill}           
 \caption{True positive and true negative rates averaged over 100 replications with $p=200$ for Model 1 in (a) and (c); and for Model 2 in (b) and (d). The methods displayed are the estimator we proposed in Section \ref{proposed} (dashed and $\blacksquare$), the method proposed by \citet{guo2010simultaneous} (dots and $\fullmoon$), and the method proposed by \citet{mai2015multiclass} (dots and $\triangle$).}
   \label{fig:VS_mod1_2}
  \end{figure}

  We could also have selected tuning parameters for the model-based methods 
  by maximizing a validation likelihood or using an information criterion, 
  but minimizing the misclassification rate on a validation set made it 
  fairer to compare the model-based methods and the Fisher-linear-discriminant-based 
  methods in terms of classification accuracy.

  \subsection{Performance measures}
  We measured classification accuracy on the test set for each replication for the methods described in Section \ref{sec:methods}. 
  For the methods that produced a precision matrix estimator, we also measured this estimator's Frobenius norm error:
  $ \| \bar{\Omega} - \Omega_*\|_F$, where $\bar{\Omega}$ is the estimator. To measure variable selection accuracy, 
  we used both the true positive rate and the true negative rate, which are respectively defined by 
  $$ \dfrac{{\rm card} \left\{(m,k): \hat{\Delta}_{m, k} 
  \neq 0 \cap \Delta_{*m, k} \neq 0 \right\}}{ 
  {\rm card}\left\{(m,k):\Delta_{*m, k} \neq 0 \right\}}, 
  \quad\dfrac{{\rm card} \left\{(m,k): \hat{\Delta}_{m, k} 
  = 0 \cap \Delta_{*m, k} = 0 \right\}}{ 
  {\rm card}\left\{(m,k):\Delta_{*m, k} = 0 \right\}},$$
  where $(m, k) \in \left\{2, \dots, J\right\} \times \left\{ 1,\dots, p\right\}$, $\Delta_{*m} = \beta_{*1} - \beta_{*m}$, $\hat{\Delta}_{m}$ is an estimator of $\Delta_{*m}$,
  and ${\rm card}$ denotes the cardinality of a set.

  \subsection{Results}
  We display average misclassification rates and Frobenius norm error averages for both models with $p=200$ in 
  Figure \ref{fig:misclass_mod1_2_later}, and display variable selection accuracy averages  in Figure \ref{fig:VS_mod1_2}. 
  For both models, our method outperformed all competitors in terms of classification 
  accuracy for all $J$, except the Bayes rule, which uses population parameter values unknown in practice. 
  In terms of precision matrix estimation, 
  for Model 1, our method did better than the $L_1$-penalized Gaussian 
  likelihood precision matrix estimator 
  when the sample size was small, but became worse than the $L_1$-penalized Gaussian 
  likelihood precision matrix estimator as the sample size increased. 
  For Model 2, our method was worse than the $L_1$-penalized Gaussian likelihood 
  precision matrix estimator in Frobenius norm error for precision matrix estimation, but 
  was better in terms of classification accuracy. 

  The precision matrix estimation results are consistent with our theory. For Model 1, as $J$ increases, the number of nonzero entries in $A \Omega_* B$ also increases. Our convergence rate bound for estimating $\Omega_*$ gets worse as the number of nonzero entries in $A\Omega_*B$ increases, which may explain why our estimator's Frobenius norm error gets worse as $J$ increases. Because the sample size increases with $J$ and the number of nonzero entries in $\Omega_*$ is fixed for all $J$, it is expected that the $L_1$-penalized Gaussian likelihood estimator improves as $J$ increases. For Model 2, as $J$ increases, more entries in $\Omega_*$ become close to zero so the Frobenius norm error for any shrinkage estimator should decrease. However, an important point is illustrated by Model 2: although the $L_1$-penalized Gaussian likelihood estimator is more accurate in terms of estimating the precision matrix, our proposed estimator still performs better in terms of classification accuracy.

  In terms of variable selection, our method was competitive with 
  the methods proposed by \citet{guo2010simultaneous} and \citet{mai2015multiclass}. 
  For Model 1, our method tended to have a higher average true negative rate than the 
  method of \citet{guo2010simultaneous} and a lower average true positive rate than 
  the method of \citet{mai2015multiclass}. For Model 2, all methods tended to have 
  relatively high average true positive rates, while our method had a higher average 
  true negative rate than the method of \citet{mai2015multiclass}.  Although the method proposed by \citet{guo2010simultaneous} had a higher average true negative rate for Model 2 than our proposed method had, our method performed better in terms of classification accuracy. 

  The performance of the method proposed by \citet{mai2015multiclass} can be partially explained by their method's variable selection properties. Their method either includes or excludes variables for discriminating between all pairwise response category comparisons. As $J$ increases, many variables are informative for only a small number of pairwise comparisons. Their method's low true negative rate for Model 2 suggests that it selects a large number of uninformative variables as $J$ increases.

  \begin{figure}[t!]
  \centerline{\hfill
      \makebox[2.6in]{(a) }
      \hfill
    \makebox[2.6in]{(b)}
      \hfill}  
  \centerline{\hfill
      \includegraphics[width=2.6in]{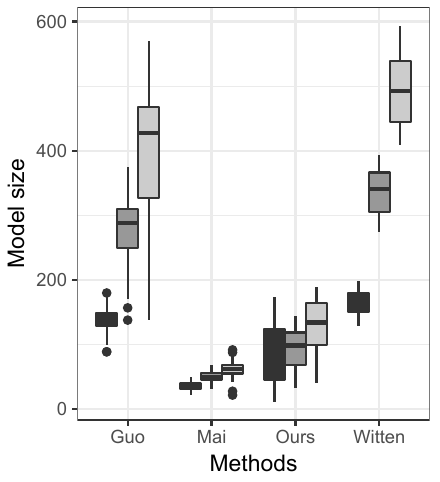}
      \hfill
     \includegraphics[width=2.6in]{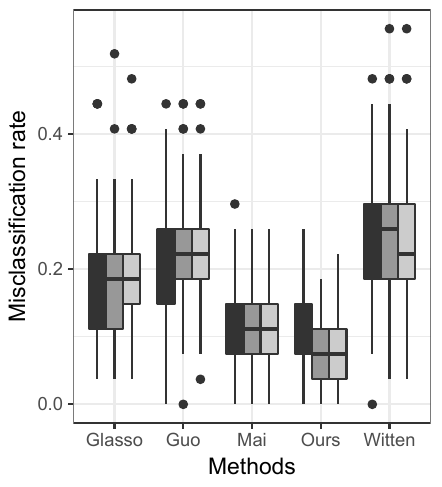}
      \hfill}  
      \caption{Model sizes and misclassification rates from 100 random training/testing splits with $k=100$ (dark grey), $k=200$ (grey), and $k=300$ (light grey). Guo is the method proposed by \citet{guo2010simultaneous}, Mai is the method proposed by \citet{mai2015multiclass}, Glasso is the $L_1$-penalized Gaussian likelihood precision matrix estimator, Ours is the estimator we propose Section \ref{proposed}, and Witten is the method proposed by \citet{witten2011penalized}.  Model size is the number of nonzero entries summed across the estimated discriminant vectors.}\label{fig:genomic_example}
  \end{figure}

  \section{Genomic data example}
  We used our method to fit the linear discriminant analysis model in a real data application. 
  The data are gene expression profiles consisting of $p=22,238$ genes from 127 subjects, 
  who either have Crohn's disease, ulcerative colitis, or neither. 
  This dataset comes from \citet{burczynski2006molecular}. 
  The goal of our analysis was to fit a linear discriminant analysis model that could be 
  used to identify which genes are informative for discriminating between each pair of the 
  response categories. 
  These data were also analyzed in \citet{mai2015multiclass}.

  To measure the classification accuracy of our method and its competitors, 
  we randomly split the data into training set of size 100 and test set of size 27 
  for 100 independent replications. Within each replication, 
  we first applied a screening rule to the training set as in \citet{rothman2009generalized} 
  and \citet{mai2015multiclass} based on $F$-test statistics, and then restricted our 
  discriminant analysis model to the genes with the $k$ largest $F$-test 
  statistic values. 

  We chose tuning parameters with five-fold 
  cross validation that minimized the validation classification error rate. Misclassification rates are shown in Figure \ref{fig:genomic_example}(b), 
  where we compared our method to those of \citet{mai2015multiclass}, of \citet{witten2011penalized}, of \citet{guo2010simultaneous}, 
  and the method that used the $L_1$-penalized Gaussian likelihood precision matrix estimator. 
  Our method was at least as accurate in terms of 
  classification as the competing methods. 
  The only method that performed nearly as well was that of \citet{mai2015multiclass} 
  with $k=100$ screened genes. However, the best out-of-sample classification 
  accuracy was achieved with $k=300,$ where our method was 
  significantly better than the competitors.

  Although the method of \citet{mai2015multiclass} 
  tended to estimate smaller models, our method, which performs best in 
  classification, selects only slightly more variables. Moreover, unlike that of \citet{mai2015multiclass}, our method can be used to 
  identify a distinct subset of genes that are informative specifically for discriminating between patients with Crohn's 
  disease and ulcerative colitis. This was of interest in the study of \citet{burczynski2006molecular}. In the Supplementary Material, we provide additional details and further results.

\section*{Acknowledgements}
We thank the associate
editor and two referees for helpful comments. 
A. J. Molstad's research was supported in part by the Doctoral Dissertation Fellowship from the University of Minnesota. 
A. J. Rothman's research is supported in part by the National Science Foundation. 
 
\section*{Supplementary Material}
Supplementary material available at \textit{Biometrika} online includes proofs of Theorem 1 and 2; and additional information about the genomic data example.

\bibliographystyle{biometrika}

\bibliography{Characteristic_Shrinkage}

\begin{thebibliography}{}

\bibitem[Bickel and Levina, 2004]{bickel2004some}
Bickel, P.~J. and Levina, E. (2004).
\newblock Some theory for {F}isher's linear discriminant function, `naive
  {B}ayes', and some alternatives when there are many more variables than
  observations.
\newblock {\em Bernoulli}, 10(6):989--1010.

\bibitem[Bickel and Levina, 2008]{bickel2008regularized}
Bickel, P.~J. and Levina, E. (2008).
\newblock Regularized estimation of large covariance matrices.
\newblock {\em The Annals of Statistics}, 36(1):199--227.

\bibitem[Boyd et~al., 2011]{boyd2011distributed}
Boyd, S., Parikh, N., Chu, E., Peleato, B., and Eckstein, J. (2011).
\newblock Distributed optimization and statistical learning via the alternating
  direction method of multipliers.
\newblock {\em Foundations and Trends in Machine Learning}, 3(1):1--122.

\bibitem[Brodie et~al., 2009]{brodie2009sparse}
Brodie, J., Daubechies, I., De~Mol, C., Giannone, D., and Loris, I. (2009).
\newblock Sparse and stable {M}arkowitz portfolios.
\newblock {\em Proceedings of the National Academy of Sciences},
  106(30):12267--12272.

\bibitem[Burczynski et~al., 2006]{burczynski2006molecular}
Burczynski, M.~E., Peterson, R.~L., Twine, N.~C., Zuberek, K.~A., Brodeur,
  B.~J., Casciotti, L., Maganti, V., Reddy, P.~S., Strahs, A., Immermann, F.,
  Spinelli, W., Schertschlag, U., Slager, A.~M., Cotreau, M.~M., and Dorner,
  A.~J. (2006).
\newblock Molecular classification of {C}rohn's disease and ulcerative colitis
  patients using transcriptional profiles in peripheral blood mononuclear
  cells.
\newblock {\em The Journal of Molecular Diagnostics}, 8(1):51--61.

\bibitem[Cai and Liu, 2011]{cai2011direct}
Cai, T. and Liu, W. (2011).
\newblock A direct estimation approach to sparse linear discriminant analysis.
\newblock {\em Journal of the American Statistical Association},
  106(496):1566--1577.

\bibitem[Chen and Teboulle, 1994]{chen1994proximal}
Chen, G. and Teboulle, M. (1994).
\newblock A proximal-based decomposition method for convex minimization
  problems.
\newblock {\em Mathematical Programming}, 64(1-3):81--101.

\bibitem[Chen et~al., 2016]{chen2016regularized}
Chen, X., Xu, M., and Wu, W.~B. (2016).
\newblock Regularized estimation of linear functionals of precision matrices
  for high-dimensional time series.
\newblock {\em IEEE Transactions on Signal Processing}, 64(24):6459--6470.

\bibitem[Dalal and Rajaratnam, 2017]{dalal2014g}
Dalal, O. and Rajaratnam, B. (2017).
\newblock Sparse gaussian graphical model estimation via alternating
  minimization.
\newblock {\em Biometrika}, 104(2):379--395.

\bibitem[Deng and Yin, 2016]{deng2016global}
Deng, W. and Yin, W. (2016).
\newblock On the global and linear convergence of the generalized alternating
  direction method of multipliers.
\newblock {\em Journal of Scientific Computing}, 66(3):889--916.

\bibitem[Fan et~al., 2012]{fan2012road}
Fan, J., Feng, Y., and Tong, X. (2012).
\newblock A road to classification in high dimensional space: the regularized
  optimal affine discriminant.
\newblock {\em Journal of the Royal Statistical Society: Series B (Statistical
  Methodology)}, 74(4):745--771.

\bibitem[Fan et~al., 2009]{fan2009network}
Fan, J., Feng, Y., and Wu, Y. (2009).
\newblock Network exploration via the adaptive {LASSO} and {SCAD} penalties.
\newblock {\em The Annals of Applied Statistics}, 3(2):521--541.

\bibitem[Fan et~al., 2016]{fan2016overview}
Fan, J., Liao, Y., and Liu, H. (2016).
\newblock An overview of the estimation of large covariance and precision
  matrices.
\newblock {\em The Econometrics Journal}, 19(1):C1--C32.

\bibitem[Friedman et~al., 2008]{friedman2008sparse}
Friedman, J.~H., Hastie, T.~J., and Tibshirani, R.~J. (2008).
\newblock Sparse inverse covariance estimation with the graphical lasso.
\newblock {\em Biostatistics}, 9:432--441.

\bibitem[Guo, 2010]{guo2010simultaneous}
Guo, J. (2010).
\newblock Simultaneous variable selection and class fusion for high-dimensional
  linear discriminant analysis.
\newblock {\em Biostatistics}, 11:599--608.

\bibitem[Lam and Fan, 2009]{lam2009sparsistency}
Lam, C. and Fan, J. (2009).
\newblock Sparsistency and rates of convergence in large covariance matrix
  estimation.
\newblock {\em Annals of Statistics}, 37(6B):4254--4278.

\bibitem[Lange, 2016]{lange2016mm}
Lange, K. (2016).
\newblock {\em MM Optimization Algorithms}.
\newblock SIAM, Philadelphia, PA.

\bibitem[Ledoit and Wolf, 2004]{ledoit2004well}
Ledoit, O. and Wolf, M. (2004).
\newblock A well-conditioned estimator for large-dimensional covariance
  matrices.
\newblock {\em Journal of Multivariate Analysis}, 88(2):365--411.

\bibitem[Mai et~al., 2018]{mai2015multiclass}
Mai, Q., Yang, Y., and Zou, H. (2018).
\newblock Multiclass sparse discriminant analysis.
\newblock {\em Statistica Sinica}.
\newblock {t}o appear, doi:10.5705/ss.202016.0117.

\bibitem[Mai et~al., 2012]{mai2012direct}
Mai, Q., Zou, H., and Yuan, M. (2012).
\newblock A direct approach to sparse discriminant analysis in ultra-high
  dimensions.
\newblock {\em Biometrika}, 99:29--42.

\bibitem[Markowitz, 1952]{markowitz1952portfolio}
Markowitz, H. (1952).
\newblock Portfolio selection.
\newblock {\em The Journal of Finance}, 7(1):77--91.

\bibitem[Mesko et~al., 2010]{mesko2010peripheral}
Mesko, B., Poliskal, S., Szegedi, A., Szekanecz, Z., Palatka, K., Papp, M., and
  Nagy, L. (2010).
\newblock Peripheral blood gene expression patterns discriminate among chronic
  inflammatory diseases and healthy controls and identify novel targets.
\newblock {\em BMC Medical Genomics}, 3(1):15.

\bibitem[Negahban et~al., 2012]{negahban2009unified}
Negahban, S.~N., Yu, B., Wainwright, M.~J., and Ravikumar, P.~K. (2012).
\newblock A unified framework for high-dimensional analysis of {$M$}-estimators
  with decomposable regularizers.
\newblock {\em Statistical Science}, 27(4):538--557.

\bibitem[Pourahmadi, 2013]{pourahmadi2013high}
Pourahmadi, M. (2013).
\newblock {\em High-Dimensional Covariance Estimation: With High-Dimensional
  Data.}
\newblock Wiley, Hoboken, NJ.

\bibitem[Price et~al., 2015]{price2015ridge}
Price, B.~S., Geyer, C.~J., and Rothman, A.~J. (2015).
\newblock Ridge fusion in statistical learning.
\newblock {\em Journal of Computational and Graphical Statistics},
  24(2):439--454.

\bibitem[Rothman, 2012]{rothman2012positive}
Rothman, A.~J. (2012).
\newblock Positive definite estimators of large covariance matrices.
\newblock {\em Biometrika}, 99:733--740.

\bibitem[Rothman et~al., 2008]{rothman2008sparse}
Rothman, A.~J., Bickel, P.~J., Levina, E., and Zhu, J. (2008).
\newblock Sparse permutation invariant covariance estimation.
\newblock {\em Electronic Journal of Statistics}, 2:494--515.

\bibitem[Rothman and Forzani, 2014]{rothman2014existence}
Rothman, A.~J. and Forzani, L. (2014).
\newblock On the existence of the weighted bridge penalized {G}aussian
  likelihood precision matrix estimator.
\newblock {\em Electronic Journal of Statistics}, 8:2693--2700.

\bibitem[Rothman et~al., 2009]{rothman2009generalized}
Rothman, A.~J., Levina, E., and Zhu, J. (2009).
\newblock Generalized thresholding of large covariance matrices.
\newblock {\em Journal of the American Statistical Association},
  104(485):177--186.

\bibitem[Taleban et~al., 2015]{taleban2015ocular}
Taleban, S., Li, D., Targan, S.~R., Ippoliti, A., Brant, S.~R., Cho, J.~H.,
  Duerr, R.~H., Rioux, J.~D., Silverberg, M.~S., Vasiliauskas, E.~A., Rotter,
  J.~I., Haritunians, T., Shih, D.~Q., Dubinsky, M., Melmed, G.~Y., and
  McGovern, D.~P. (2015).
\newblock Ocular manifestations in inflammatory bowel disease are associated
  with other extra-intestinal manifestations, gender, and genes implicated in
  other immune-related traits.
\newblock {\em Journal of Crohn's and Colitis}, 10(1):43--49.

\bibitem[Toyonaga et~al., 2016]{toyonaga2016lipocalin}
Toyonaga, T., Matsuura, M., Mori, K., Honzawa, Y., Minami, N., Yamada, S.,
  Kobayashi, T., Hibi, T., and Nakase, H. (2016).
\newblock Lipocalin 2 prevents intestinal inflammation by enhancing phagocytic
  bacterial clearance in macrophages.
\newblock {\em Scientific Reports}, 6.

\bibitem[Witten and Tibshirani, 2011]{witten2011penalized}
Witten, D.~M. and Tibshirani, R. (2011).
\newblock Penalized classification using {F}isher's linear discriminant.
\newblock {\em Journal of the Royal Statistical Society: Series B (Statistical
  Methodology)}, 73(5):753--772.

\bibitem[Witten and Tibshirani, 2009]{witten2009covariance}
Witten, D.~M. and Tibshirani, R.~J. (2009).
\newblock Covariance-regularized regression and classification for high
  dimensional problems.
\newblock {\em Journal of the Royal Statistical Society: Series B (Statistical
  Methodology)}, 71(3):615--636.

\bibitem[Xu et~al., 2015]{xu2015covariance}
Xu, P., Zhu, J., Zhu, L., and Li, Y. (2015).
\newblock Covariance-enhanced discriminant analysis.
\newblock {\em Biometrika}, 102:33--45.

\bibitem[Xue et~al., 2012]{xue2012positive}
Xue, L., Ma, S., and Zou, H. (2012).
\newblock Positive-definite $\ell_1$-penalized estimation of large covariance
  matrices.
\newblock {\em Journal of the American Statistical Association},
  107(500):1480--1491.

\bibitem[Yuan and Lin, 2007]{yuan2007model}
Yuan, M. and Lin, Y. (2007).
\newblock Model selection and estimation in the {G}aussian graphical model.
\newblock {\em Biometrika}, 93:19--35.

\end{thebibliography}

\clearpage

\title{Supplementary Material for ``Shrinking characteristics of precision matrix estimators''}
\author{Aaron J. Molstad$^1$ \hspace{2pt} and Adam J. Rothman$^2$\\
Biostatistics Program, Fred Hutchinson Cancer Research Center$^1$\\
  School of Statistics, University of Minnesota$^2$\\
  \texttt{amolstad@fredhutch.org}$^1$ \hspace{5pt} \texttt{arothman@umn.edu}$^2$}
\date{}

\maketitle
\section{Proofs}
\subsection{Notation}
  Define the following norms: $\|A\|_{\infty} = \max_{i,j} |A_{ij}|,$
   $|A|_1 = \sum_{i,j} |A_{ij}|$, and $\|A\|_F = {\rm tr}(A^\T A)$. Let $\mathbb{S}^{p_n}$ 
   denote the set of $p_n \times p_n$ symmetric matrices. To simplify notation, let $\kappa = k_1^{-2}$. 

  \subsection{Proof of Theorem 1}
    To prove Theorem 1, we use a strategy similar to that employed by \citet{rothman2012positive}. 
  \begin{lemma}\label{lemma1}
  Under Assumptions 1--3, 
  if $\lambda_n \leq \epsilon \kappa \left\{ \xi(p_n, \mathcal{G}) \eta_1\right\}^{-1}$ for some 
  $\eta_1 > 12,$ then for all positive and sufficiently small
  $\epsilon$, $\| B^{+}(S_n - \Omega_*^{-1})A^{+}\|_{\infty} \leq \lambda_n/2$ 
  implies $\|\hat{\Omega} - \Omega_*\|_F \leq \epsilon$ for sufficiently large $n$. 
  \end{lemma}

 \begin{proof} We follow the proof techniques 
  used by \citet{rothman2008sparse}, \citet{negahban2009unified} and \citet{rothman2012positive}. 
  Define $B_\epsilon = \left\{ \Delta \in \mathbb{S}^{p_n}: \|\Delta\|_F = \epsilon \right\}.$ Let $f$ be the objective function in (2). 
  Because $f$ is convex and $\hat{\Omega}$ is its minimizer,
  ${\rm inf}\left\{ f(\Omega_* + \Delta): \Delta \in B_\epsilon \right\} > f(\Omega_*),$ 
  implies $\|\hat{\Omega} - \Omega_*\|_F \leq \epsilon$ \citep{rothman2008sparse}.  Define  $D(\Delta) = f(\Omega_* + \Delta) - f(\Omega_*).$ Then
  \begin{equation}
  D(\Delta) = {\rm tr}(S_n \Delta)  + \log {\rm det}(\Omega_*)
   - \log{\rm det}(\Omega_* + \Delta) + \lambda_n\left\{ |A(\Omega_* + \Delta)B|_1 - |A\Omega_*B|_1\right\}.\notag
   \end{equation}
  By the arguments used in \citet{rothman2008sparse}, $ \log {\rm det}(\Omega_*)
   - \log{\rm det}(\Omega_* + \Delta) \geq -{\rm tr}(\Omega_*^{-1}\Delta) + 8^{-1}\kappa \|\Delta\|_F^2,$
   so that
  \begin{equation}
  D(\Delta)  \geq {\rm tr}\left\{ \Delta(S_n - \Omega_*^{-1})\right\} 
  + \frac{1}{8}\kappa \|\Delta\|_F^2 + 
  \lambda_n\left\{ |A(\Omega_* + \Delta)B|_1 - |A\Omega_*B|_1\right\}. \label{D(Delta)2}
  \end{equation}
  We now bound $|A(\Omega_* + \Delta)B|_1 - |A\Omega_*B|_1$ in \eqref{D(Delta)2}. Recall that 
  $$\mathcal{G} = \left\{(i,j) \in \left\{1, \dots, a\right\} \times 
  \left\{1, \dots, b\right\}: [A\Omega_* B]_{ij} \neq 0 \right\}$$ and 
  $\mathcal{G}^{c} = \left\{1, \dots, a\right\} \times \left\{1, \dots, b\right\} \setminus \mathcal{G}$. 
  Since $|A\Omega_* B|_1  = | \left[ A \Omega_* B \right]_{\mathcal{G}} |_1$ and
  $|A(\Omega_* + \Delta)B|_1  = |\left[A\Omega_*B\right]_\mathcal{G} 
  + \left[A\Delta B\right]_{\mathcal{G}}|_1 + | \left[ A \Delta B \right]_{\mathcal{G}^c}|_1,$
  we can apply the reverse triangle inequality: 
  $|A(\Omega_* + \Delta)B|_1 - |A\Omega_*B|_1  \geq 
  |\left[A \Delta B \right]_{\mathcal{G}^c}|_1 - |\left[A \Delta B\right]_\mathcal{G}|_1. $
  Plugging this bound into \eqref{D(Delta)2}, 
  \begin{equation}
  D(\Delta) \geq  {\rm tr}\left\{(S_n - \Omega_*^{-1}) \Delta \right\} + \frac{1}{8}\kappa \|\Delta\|_F^2 +
   \lambda_n \left( |\left[A \Delta B \right]_{\mathcal{G}^c}|_1 - 
  |\left[A \Delta B\right]_\mathcal{G}|_1  \right). \label{D(Delta)Lemma2}
  \end{equation}
  We now bound ${\rm tr}\left\{(S_n - \Omega_*^{-1}) \Delta \right\}$. 
  Let $A^{+} = (A^\T A)^{-1}A^\T$ and $B^{+} = B^\T (BB^\T)^{-1}$. Because $A$
   and $B$ are both rank $p_n$ by Assumption 2, $A^{+}A = I_{p_n}$ and $BB^{+} = I_{p_n}$. Thus
   \begin{align}
  {\rm tr}\left\{(S_n - \Omega_*^{-1}) \Delta \right\} &\geq 
  - |{\rm tr}\left\{(S_n - \Omega_*^{-1})  \Delta \right\}| = 
  - |{\rm tr}\left\{(S_n - \Omega_*^{-1})A^{+}A  \Delta B B^{+} \right\}| \notag \\
   & = - |{\rm tr}\left\{B^{+}(S_n - \Omega_*^{-1})A^{+}A  \Delta B  \right\}| \notag \\
   & \geq -  \|B^{+}(S_n - \Omega_*^{-1})A^{+}\|_{\infty} |A\Delta B|_1. \label{tr(SOmega)} 
   \end{align}
    By assumption, $\|B^{+}(S_n - \Omega_*^{-1})A^{+}\|_{\infty} \leq \lambda_n/2$, so applying \eqref{tr(SOmega)} to \eqref{D(Delta)Lemma2}, 
  \begin{align}
  D(\Delta) &\geq  \frac{1}{8}\kappa \|\Delta\|_F^2 - \frac{\lambda_n}{2} | A\Delta B |_1 
    + \lambda_n \left( |\left[A\Delta B\right]_{\mathcal{G}^c}|_1 
  - |\left[A \Delta B \right]_{\mathcal{G}}|_1 \right) \label{D(Delta)forLemma5} \\
  & =   \frac{1}{8}\kappa \|\Delta\|_F^2 - \frac{\lambda_n}{2}
  \left( |\left[A\Delta B\right]_\mathcal{G}|_1 
  + |\left[A \Delta B \right]_{\mathcal{G}^c}|_1 \right)
    + \lambda_n \left( |\left[A\Delta B\right]_{\mathcal{G}^c}|_1 
  - |\left[A \Delta B \right]_{\mathcal{G}}|_1 \right) \notag \\
   &=  \frac{1}{8}\kappa \|\Delta\|_F^2 - \frac{3\lambda_n}{2}
  |\left[A\Delta B\right]_\mathcal{G}|_1  + \frac{\lambda_n}{2} |\left[A\Delta B\right]_{\mathcal{G}^c}|_1 \notag \\
   &\geq  \frac{1}{8}\kappa \|\Delta\|_F^2 - \frac{3\lambda_n}{2}
  |\left[A\Delta B\right]_\mathcal{G}|_1.  \label{D(Delta)3}
  \end{align}
    Finally, we bound the quantity $|\left[A \Delta B\right]_{\mathcal{G}}|_1$. 
    Multiplying and dividing $\Delta$ by $\|\Delta\|_F$,
    $$\left|\left[A\frac{\|\Delta\|_F}{\|\Delta\|_F} \Delta B \right]_{\mathcal{G}}\right|_1 
    = \|\Delta\|_F \left|\left[A \frac{1}{\|\Delta\|_F} \Delta B \right]_\mathcal{G}\right|_1 
    \leq \|\Delta\|_F \left( \sup_{M \in \mathbb{S}^p, M \neq 0} \frac{|\left[A M B \right]_\mathcal{G}|_1}{\|M\|_F} \right),$$ 
    so that $|\left[A \Delta B\right]_{\mathcal{G}}|_1 \leq \|\Delta\|_F \hspace{2pt} \xi(p_n, \mathcal{G})$.
    Hence, because $\|\Delta\|_F = \epsilon$ when $\Delta \in B_\epsilon$, if $\lambda_n \leq\epsilon \kappa\left\{\xi(p_n, \mathcal{G})\eta _1\right\}^{-1}$ 
    with $\eta_1 > 12$,
    \begin{align*}
     D(\Delta) & \geq \frac{1}{8}\kappa \|\Delta\|_F^2 
     - \frac{3\lambda_n}{2}  \|\Delta\|_F \xi(p_n, \mathcal{G}) \\ 
     & = \|\Delta\|_F^2 \left\{ \frac{1}{8}\kappa 
     - \frac{3 \lambda_n \xi(p_n, \mathcal{G})}{2\|\Delta\|_F} \right\}
     \geq  \epsilon^2 \left(\frac{1}{8}\kappa - \frac{3}{2\eta_1}\kappa \right) > 0. 
  \end{align*} 
  which establishes the desired result. 
  \end{proof}

The following lemma follows from the proof of Lemma 1 of \citet{negahban2009unified}. 
  \begin{lemma}\label{lemma2} Under the conditions of Lemma~\ref{lemma1}, $\hat{\Delta} = \hat{\Omega} - \Omega_*$ 
  belongs to the set 
  $$\left\{ \Delta \in \mathbb{S}^{p_n}: |\left[A \Delta B \right]_{\mathcal{G}^c}|_1 
  \leq 3 |\left[A \Delta B \right]_{\mathcal{G}}|_1\right\}.$$
  \end{lemma}

Lemma \ref{lemma3} follows from the proof of Lemma 2 from \citet{lam2009sparsistency}, Assumption 2, and Lemma A.3 of \citet{bickel2008regularized}.

  \begin{lemma}\label{lemma3} Under Assumptions 1--3, there exist constants $C_1$ and $C_2$ such that 
  $$ {\rm pr}(\|B^{+}S_n A^{+} - B^{+}\Omega_*^{-1}A^{+}\|_{\infty} \geq \nu) \leq C_1 p_n^2 {\rm exp}(-C_2 n \nu^2),$$
  for $|\nu| \leq \delta$ where $C_1, C_2,$ and $\delta$ do not depend on $n$. 
  \end{lemma}

\begin{proof}[Proof of Theorem 1]
  Set $\epsilon = \kappa^{-1}\eta_1 \xi(p_n, \mathcal{G})K_1  (n^{-1} \log p_n)^{1/2}$ and  $\lambda_n = K_1( n^{-1} \log p_n)^{1/2}$ with $\eta_1 > 12$.  Applying Lemma 1 and Lemma 3, there exist constants $C_1$ and $C_2$ such that for sufficiently large $n,$ 
  ${\rm pr}\{\|\hat{\Omega} - \Omega_*\|_F \leq \kappa^{-1}\eta_1 \xi(p_n, \mathcal{G}) K_1 (n^{-1}\log p_n)^{1/2}\}  \geq {\rm pr}\{ \|B^{+}(S_n - \Omega_*^{-1})A^{+}\|_{\infty} \leq K_1 (n^{-1}\log p_n)^{1/2}/2 \} \geq 1 - C_1p_n^{2 - C_2 K_1^2},$
  which establishes $(i)$ because $1 - C_1p_n^{2 - C_2 K_1^2} \to 1$ as $K_1 \to \infty$. To establish $(ii)$, 
  \begin{align}
  |A(\hat{\Omega} - \Omega_*)B|_1 & = | [A(\hat{\Omega} - \Omega_*)B]_{\mathcal{G} }|_1 + |[A (\hat{\Omega} - \Omega_*)B]_{\mathcal{G}^c} |_1 \notag \\
  & \leq 4  | [A(\hat{\Omega} - \Omega_*)B]_{\mathcal{G} }|_1 \label{Thm1Prf2}\\
  &  \leq 4 \xi(p_n, \mathcal{G}) \|\hat{\Omega} - \Omega_*\|_F, \label{Thm1Prf2_1}
  \end{align}
  where \eqref{Thm1Prf2} follows from Lemma \ref{lemma2} and \eqref{Thm1Prf2_1} follows from the definition of $\xi(p_n, \mathcal{G})$. 
  \end{proof}

  \subsection{Proof of Theorem 2}

  \begin{lemma}\label{lemma4} Let $M_a$ and $M_b$ be positive constants; and let $a_n$ and $b_n$ be positive sequences. Define $\mathcal{D}_n(M_a, M_b)$ as the event that
  $|(\hat{A}_n - A)A^{+}|_1 \leq M_a a_n$ and $|B^{+}(\hat{B}_n - B)|_1 \leq M_b b_n$. Then   
  \begin{align*} |\hat{A}_n(\Omega_* + \Delta)\hat{B}_n|_1 - |\hat{A}_n\Omega_*\hat{B}_n|_1 
  \geq |A(\Delta + \Omega_*)B|_1 & - |[A \Omega_* B]_\mathcal{G}|_1 \\ & -  \tilde{M}_n \left( |A\Delta B|_1 
  + 2| [A\Omega_* B]_\mathcal{G}|_1\right)
  \end{align*}
  on $\mathcal{D}_n(M_a, M_b)$, where $\tilde{M}_{n} = M_a a_n + M_b b_n + M_a M_b a_n b_n.$
  \end{lemma}

 \begin{proof} Let
  $|\hat{A}_n(\Omega_* + \Delta)\hat{B}_n|_1 - |\hat{A}_n\Omega_*\hat{B}_n|_1 \equiv V_1 - V_2.$ 
  First,
  \begin{align}
  V_1  & =  |\hat{A}_n (\Omega_* + \Delta)\hat{B}_n + A(\Omega_* + \Delta)B - A(\Omega_* + \Delta)B|_1 \notag \\
   & \geq  |A(\Omega_* + \Delta)B |_1 - |A(\Omega_* + \Delta)B - \hat{A}_n(\Omega_* + \Delta)\hat{B}_n|_1, \label{V1}
  \end{align}
by the triangle inequality. Also,
  \begin{equation}
  V_2 = |\hat{A}_n\Omega_*\hat{B}_n - A\Omega_* B + A\Omega_* B|_1 
  \leq |\hat{A}_n\Omega_*\hat{B}_n - A\Omega_*B|_1 + |[A\Omega_* B]_\mathcal{G}|_1, \label{V2}
  \end{equation}
  so that from \eqref{V1} and \eqref{V2},  
  \begin{align}
  V_1 - V_2  \geq & |A(\Omega_* + \Delta)B |_1 - |[A\Omega_* B]_\mathcal{G}|_1 \notag \\
  &    - |A(\Omega_* + \Delta)B - \hat{A}_n(\Omega_* + \Delta)\hat{B}_n|_1 
  - |\hat{A}_n\Omega_*\hat{B}_n - A\Omega_*B |_1. \label{V1-V2}
  \intertext{Let $V_3 = -|A(\Omega_* + \Delta)B - \hat{A}_n(\Omega_* + \Delta)\hat{B}_n|_1 
  - |\hat{A}_n\Omega_*\hat{B}_n - A\Omega_*B|_1$. By a triangle inequality on the first term of $V_3$, 
    }
   V_3 &\geq - 2| \hat{A}_n \Omega_* \hat{B}_n - A \Omega_* B|_1 
  - | \hat{A}_n \Delta \hat{B}_n - A \Delta B|_1. \label{V3}
  \end{align}
 To bound \eqref{V3}, we need to bound functions of the 
    form $|AX B - \hat{A}_n X \hat{B}_n|_1; $
    \begin{align}
    |AXB  - \hat{A}_nX \hat{B}_n|_1 =& | (A - \hat{A}_n)X B + A X (B - \hat{B}_n) 
    +(A - \hat{A}_n)X (\hat{B}_n - B)|_1 \notag \\
     \leq&  |(A - \hat{A}_n)X B|_1 + |A X (B - \hat{B}_n)|_1 
     + |(A - \hat{A}_n)X (\hat{B}_n - B)|_1.\label{triange1}\\
     =& |(A - \hat{A}_n)A^{+} AX B|_1 + |A X B B^{+}(B - \hat{B}_n)|_1 \notag \\
     & \hspace{104pt}+ |(A - \hat{A}_n)A^{+}AX B B^{+} (\hat{B}_n - B)|_1, \label{A3step}\\
     \leq & |A X B|_1 \left\{|(A - \hat{A}_n)A^{+}|_1 + |B^{+}(B - \hat{B}_n)|_1 \right. \notag \\
     & \hspace{123pt}\left.+ |(A - \hat{A}_n)A^{+}|_1  |B^{+}(B - \hat{B}_n)|_1\right\} \label{submulti}\\
     \leq&  |A X B|_1 \left(M_a a_n  + M_b b_n + M_a M_b a_n b_n \right) , \label{last step}
     \end{align}
      where  \eqref{triange1} follows from the triangle inequality; 
      \eqref{A3step} follows from Assumption 2 and the definition of $A^{+}$ and $B^{+}$; 
     \eqref{submulti} follows from the sub-multiplicative property 
     of the $| \cdot|_1$ norm, and  \eqref{last step} occurs on $\mathcal{D}_n(M_a, M_b).$
     Applying \eqref{last step} to both terms in \eqref{V3} gives
    $$ V_3 \geq - 2| \hat{A}_n \Omega_* \hat{B}_n - A \Omega_* B|_1 
    - | \hat{A}_n \Delta \hat{B}_n - A \Delta B|_1 \geq -\tilde{M_n} 
    \left( 2|[A \Omega_* B]_\mathcal{G}|_1 +  |A\Delta B|_1\right). $$ 
   Plugging this bound into \eqref{V1-V2} gives the result. 
   \end{proof}

  \begin{lemma}\label{lemma5} Under Assumptions 1--3 and 5,
  if $n$ is sufficiently large and $\lambda_n \leq \epsilon \kappa (T_{\epsilon} \eta_2)^{-1}$ 
  for some $\eta_2 > 8$, where  
  $$T_\epsilon = \left\{ \left(\frac{3}{2} + \tilde{M}_n \right)\xi(p_n, \mathcal{G}) 
  + 2\tilde{M}_n \frac{|[A \Omega_* B]_\mathcal{G} |_1}{\epsilon} \right\},$$ 
  then for all positive and sufficiently small $\epsilon$, $\| B^{+}(S_n - \Omega_*^{-1})A^{+}\|_{\infty} 
  \leq \lambda_n/2$, implies $\|\tilde{\Omega} - \Omega_*\|_F \leq \epsilon$.
  \end{lemma}

  \begin{proof} Let $\tilde{f}$ be the objective function from (9). Define 
  $\tilde{D}(\Delta) = \tilde{f}(\Omega_* + \Delta) - \tilde{f}(\Omega_*)$ so that 
  $$ \tilde{D}(\Delta)  = {\rm tr}(S_n\Delta) + \log {\rm det}(\Omega_*) 
  - \log{\rm det}(\Omega_* + \Delta) 
  + \lambda_n \left\{ |\hat{A}_n(\Omega_* + \Delta)\hat{B}_n|_1 
  - |\hat{A}_n\Omega_*\hat{B}_n|_1\right\}.$$ 
    As in the proof of Lemma \ref{lemma1}, 
  we want to show 
 ${\rm inf}\{\tilde{f}(\Omega_* + \Delta): \Delta \in B_\epsilon \} > \tilde{f}(\Omega_*)$. By Assumption 5, for sufficiently large $n$, the bound in Lemma 4 holds with probability arbitrarily close to one for sufficiently large constants $M_a$ and $M_b$. Thus, let $\tilde{M}_n =  a_n M_a + b_n M_b + a_n b_n M_a M_b$ with $M_a$ and $M_b$ sufficiently large. Then, using the bound in Lemma 4 and applying the same arguments as in the proof of Lemma \ref{lemma1} to obtain \eqref{D(Delta)forLemma5}, 
  \begin{align}
  \tilde{D}(\Delta) \geq& \frac{1}{8}\kappa \|\Delta\|_F^2 
  -  \frac{\lambda_n}{2} \left(|\left[A \Delta B\right]_\mathcal{G} |_1 
  + |\left[A \Delta B\right]_{\mathcal{G}^c}|_1 \right)
   + \lambda_n \left( |\left[A \Delta B\right]_{\mathcal{G}^c} |_1 
   - |\left[A \Delta B\right]_{\mathcal{G}} |_1\right) \notag \\
  &  - \tilde{M_n} \lambda_n \left( |A \Delta B|_1 + 2 |[A \Omega_* B]_\mathcal{G}|_1\right) \notag \\
   = &  \frac{1}{8}\kappa \|\Delta\|_F^2 -  \frac{3\lambda_n}{2} |\left[A \Delta B\right]_\mathcal{G} |_1  
   + \frac{\lambda_n}{2} |[A \Delta B]_{\mathcal{G}^c}|_1 - \tilde{M}_n \lambda_n \left( |A \Delta B|_1  
   + 2|[A \Omega_* B]_\mathcal{G}|_1\right) \notag \\
   = & \frac{1}{8}\kappa \|\Delta\|_F^2 -  \frac{3\lambda_n}{2}  |\left[A \Delta B\right]_\mathcal{G} |_1  
   + \frac{\lambda_n}{2} |[A \Delta B]_{\mathcal{G}^c}|_1 \notag \\
   &  - \tilde{M}_n \lambda_n \left( |[A \Delta B]_\mathcal{G}|_1 
   + |[A \Delta B]_{\mathcal{G}^c}|_1  + 2|[A \Omega_* B]_\mathcal{G}|_1\right) \notag \\
   = & \frac{1}{8}\kappa \|\Delta\|_F^2 -  \left(\frac{3}{2} 
   + \tilde{M}_n\right)\lambda_n |\left[A \Delta B\right]_\mathcal{G} |_1  \notag \\ 
   & 
   + \left(\frac{1}{2} - \tilde{M}_n\right) \lambda_n |[A \Delta B]_{\mathcal{G}^c}|_1 
   - 2 \tilde{M}_n \lambda_n |[A \Omega_* B]_\mathcal{G}|_1 \label{lastline_lemma5}
   \end{align}
  and because $\tilde{M}_n = o(1)$ by Assumption 5, for sufficiently large $n$, \eqref{lastline_lemma5} implies
  \begin{align}
  \tilde{D}(\Delta) \geq&  \frac{1}{8}\kappa \|\Delta\|_F^2 
  - \left(\frac{3}{2} + \tilde{M}_n\right)\lambda_n |\left[A \Delta B\right]_\mathcal{G} |_1  
  - 2 \tilde{M}_n \lambda_n |[A \Omega_* B]_\mathcal{G}|_1\notag \\
   \geq &\|\Delta\|_F^2 \left\{ \frac{1}{8}\kappa  
  - \left(\frac{3}{2} + \tilde{M}_n \right)\frac{\lambda_n}{\|\Delta\|_F}\xi(p_n, \mathcal{G}) 
  - 2\tilde{M}_n \lambda_n \frac{|[A \Omega_* B]_\mathcal{G} |_1}{\|\Delta\|_F^2}  \right\}\notag \\
   = & \|\Delta\|_F^2 \left[ \frac{1}{8}\kappa  
  -\frac{\lambda_n}{\|\Delta\|_F} \left\{ \left(\frac{3}{2} + \tilde{M}_n \right)
  \xi(p_n, \mathcal{G}) + 2\tilde{M}_n \frac{|[A \Omega_* B]_\mathcal{G} |_1}{\|\Delta\|_F} \right\} \right].
  \label{last_line_mid}
  \end{align}
  Since $\|\Delta\|_F = \epsilon$ when $\Delta \in B_\epsilon$, if $\lambda_n \leq  \epsilon \kappa( T_\epsilon \eta_2)^{-1}$ 
  for some $\eta_2 > 8$, where 
  $$ T_\epsilon = \left\{ \left(\frac{3}{2} 
  + \tilde{M}_n \right)\xi(p_n, \mathcal{G}) 
 + 2\tilde{M}_n \frac{|[A \Omega_* B]_\mathcal{G} |_1}{\epsilon} \right\},$$
  the inequality from \eqref{last_line_mid} implies 
  $$\tilde{D}(\Delta)  \geq \epsilon^2 \left[ \frac{1}{8}\kappa  
  -\frac{\lambda_n}{\epsilon} \left\{ \left(\frac{3}{2} 
  + \tilde{M}_n \right)\xi(p_n, \mathcal{G}) 
 + 2\tilde{M}_n \frac{|[A \Omega_* B]_\mathcal{G} |_1}{\epsilon} \right\} \right]
  \geq \epsilon^2  \left( \frac{1}{8}\kappa  - \frac{1}{\eta_2}\kappa \right) > 0, $$
  which establishes the desired result. 
  \end{proof}

  \begin{lemma}\label{lemma6} Under the conditions of Lemma 5, 
  $\tilde{\Delta} = \tilde{\Omega} - \Omega_*$ 
  belongs to the set $$\left\{ \Delta \in  \mathbb{S}^{p_n}: |[A \Delta B]_{\mathcal{G}^c}|_1 \leq 
  \frac{(3 + 2\tilde{M}_n) |[A \Delta B]_\mathcal{G}|_1  + 
  4 \tilde{M}_n |[A \Omega_* B]_\mathcal{G}|_1}{1 - 2 \tilde{M}_n}\right\}.$$
  \end{lemma}

\begin{proof}
  Using the same arguments as in the proof of Lemma 1 from 
  \citet{negahban2009unified}, and from \eqref{lastline_lemma5}, we have  
  \begin{align*}
  0 \geq \tilde{D}(\tilde{\Delta})  \geq & - \frac{\lambda_n}{2}(1 + 2\tilde{M}_n)
  \left( |[A \tilde{\Delta} B]_\mathcal{G}|_1 + |[A \tilde{\Delta} B]_{\mathcal{G}^c}|_1\right) 
  + \lambda_n \left( |[A \tilde{\Delta} B]_{\mathcal{G}^c}|_1 - |[A \tilde{\Delta} B]_{\mathcal{G}}|_1\right) \\
  & - 2\lambda_n \tilde{M}_n  |[A \Omega_* B]_\mathcal{G}|_1 \\
   = & -\frac{\lambda_n}{2}\left\{(3 + 2 \tilde{M}_n) |[A \tilde{\Delta} B]_\mathcal{G}|_1 
  -  (1 - 2 \tilde{M}_n) |[A \tilde{\Delta} B]_{\mathcal{G}^c}|_1 +  4\tilde{M}_n  |[A \Omega_* B]_\mathcal{G}|_1\right\} 
  \end{align*}
  so that 
   $$ |[A \tilde{\Delta} B]_{\mathcal{G}^c}|_1 \leq \frac{(3 + 2 \tilde{M}_n) |[A \tilde{\Delta} B]_\mathcal{G}|_1  
   + 4 \tilde{M}_n |[A \Omega_* B]_\mathcal{G}|_1}{1 - 2 \tilde{M}_n},$$
   which is the desired inequality.
\end{proof}

  \begin{proof}[Proof of Theorem 2]
  Set $\lambda_n = K_2 (n^{-1} \log p_n)^{1/2}$ and $\epsilon = \lambda_n T_{\epsilon}\eta_2\kappa^{-1}$ for some $\eta_2 > 8$. We can simplify the expression for $\epsilon$ by solving 
  $$\epsilon \kappa \eta^{-1}_2 (K_2^2   n^{-1}\log p_n)^{-1/2} = 
  \left\{ \left(\frac{3}{2} + \tilde{M}_n \right)\xi(p_n, \mathcal{G}) 
  + 2\tilde{M}_n \frac{|[A \Omega_* B]_\mathcal{G} |_1}{\epsilon} \right\} ,$$
   or equivalently,
  \begin{align}
  \epsilon^2 \kappa \eta^{-1}_2 (K_2^2 n^{-1}\log p_n)^{-1/2}  
  - \epsilon \left( \frac{3}{2} + \tilde{M}_n \right) \xi(p_n, \mathcal{G}) 
  - 2 \tilde{M}_n |[A \Omega_* B]_\mathcal{G}|_1 = 0.  \label{quadratic_epsilon}
  \end{align}
  Using the quadratic formula to solve \eqref{quadratic_epsilon} for $\epsilon$,
  \begin{align} \epsilon = & \frac{K_2\eta_2}{2\kappa } \left(\frac{\log p_n}{n}\right)^{1/2}
  \left[ \left( \frac{3}{2} + \tilde{M}_n \right) \xi(p, \mathcal{G}) \right. \notag \\
   & \quad \quad \left. +  \left\{  \left( \frac{3}{2} + \tilde{M}_n \right)^2 \xi^2(p_n, \mathcal{G}) 
  + \frac{8 \tilde{M}_n \kappa }{K_2 \eta_2}\left(\frac{n}{\log p_n}\right)^{1/2}
  |[A \Omega_* B]_\mathcal{G}|_1 \right\}^{1/2}\right]. \label{epsilon}
  \end{align}
  To further simplify the result, we find an $\tilde{\epsilon}$ such that $\epsilon \leq \tilde{\epsilon}$. Then
  $\|\tilde{\Omega} - \Omega_*\|_F  \leq \epsilon$ implies $\|\tilde{\Omega} - \Omega_*\|_F  \leq \tilde{\epsilon}$, 
  so $\|B^{+}(S_n - \Omega_*)A^{+}\|_{\infty} \leq \lambda_n/2$ 
  also implies $\|\tilde{\Omega} - \Omega_*\|_F  \leq \tilde{\epsilon}$. 
  Viewing the square root in \eqref{epsilon} as the Euclidean norm of the sum 
  of the square root of its two terms, we use the triangle inequality to obtain 
  \begin{equation}
   \epsilon \leq \frac{K_2 \eta_2}{ \kappa} \left(\frac{\log p_n}{n}\right)^{1/2}
  \left[  \left( \frac{3}{2} + \tilde{M}_n \right) \xi(p_n, \mathcal{G}) 
  +   \left\{ \frac{2 \tilde{M}_n \kappa}{K_2 \eta_2}\left(\frac{n}{\log p_n}\right)^{1/2}
  |[A \Omega_* B]_\mathcal{G}|_1 \right\} ^{1/2}\right] \equiv \tilde\epsilon\notag.
  \end{equation}
  Then, applying Lemma \ref{lemma5} and Lemma 3, there exist constants $C_3$ and $C_4$ such that for sufficiently large $n,$
  \begin{align*}{\rm pr}\left(\|\tilde{\Omega}  - \Omega_* \|_F \leq \tilde{\epsilon}\right) & 
   \geq {\rm pr}\left\{ \|B^{+}S_n A^{+} - B^{+}\Omega_*^{-1}A^{+}\|_{\infty} \leq \frac{K_2}{2}  \left(\frac{\log p_n}{n}\right)^{1/2} \right\}\\ 
& \geq 1 - C_3 p_n^{2 - C_4 K_2^2}
   \end{align*}
  which establishes $(i)$ because $1 - C_3 p_n^{2 - C_4 K_2^2} \to 1$ as $K_2 \to \infty$. 
  To establish $(ii)$, we bound $|\hat{A}_n \tilde{\Omega} \hat{B}_n - A \Omega_*B|_1.$ 
    By the triangle inequality,
  \begin{align}|\hat{A}_n \tilde{\Omega} \hat{B}_n - A \Omega_*B|_1 
  & \leq |\hat{A}_n \tilde{\Omega} \hat{B}_n -  A\tilde{\Omega} B|_1 
  + | A \tilde{\Omega} B- A \Omega_*B|_1\label{corol1tri} 
  \end{align}
  and by the argument used to obtain the inequality in \eqref{last step},
  $|\hat{A}_n \tilde{\Omega} \hat{B}_n -  A\tilde{\Omega} B|_1  \leq \tilde{M}_n|A \tilde{\Omega}B|_1. $
  Using this bound on the first term in \eqref{corol1tri},
  \begin{equation}
   |\hat{A}_n \tilde{\Omega} \hat{B}_n - A \Omega_*B|_1 
   \leq \tilde{M}_n |A \tilde{\Omega} B|_1 + | A \tilde{\Omega} B- A \Omega_*B|_1 \label{corol1eqn}. 
   \end{equation}
   Then, bounding the first term in \eqref{corol1eqn}, 
   $
   |A \tilde{\Omega} B|_1   \leq |A \Omega_* B|_1 + |A \tilde{\Omega} B - A \Omega_* B |_1 $
   so that from \eqref{corol1eqn}, 
   \begin{equation} 
   |\hat{A}_n \tilde{\Omega} \hat{B}_n - A \Omega_*B|_1 
   \leq \tilde{M}_n |A \Omega_* B|_1 + (\tilde{M}_n + 1)|A \tilde{\Omega} B - A \Omega_* B |_1 .
   \label{main_them2ii}
   \end{equation}
   To bound the right term in the sum on the right hand side of \eqref{main_them2ii}, we apply Lemma 6 to $\tilde{\Delta} = \tilde{\Omega} - \Omega_*$ so that
  \begin{align}
  |A \tilde{\Omega}B - A \Omega_* B|_1 & = |[A \tilde{\Delta} B]_\mathcal{G}|_1 
  +|[A \tilde{\Delta} B]_{\mathcal{G}^c}|_1 \notag \\
  & \leq  |[A \tilde{\Delta} B]_\mathcal{G}|_1  + \frac{(3 + 2\tilde{M}_n) |[A \tilde{\Delta} B]_\mathcal{G}|_1  
  + 4 \tilde{M}_n |[A \Omega_* B]_\mathcal{G}|_1}{1 - 2 \tilde{M}_n}\notag \\
  & =  \frac{(1 - 2\tilde{M}_n)|[A \tilde{\Delta} B]_\mathcal{G}|_1  + (3 + 2\tilde{M}_n) |[A \tilde{\Delta} B]_\mathcal{G}|_1  
  + 4 \tilde{M}_n |[A \Omega_* B]_\mathcal{G}|_1}{1 - 2 \tilde{M}_n}\notag \\
  & = \frac{4 |[A \tilde{\Delta} B]_\mathcal{G}|_1 
  + 4 \tilde{M}_n |[A \Omega_* B]_\mathcal{G}|_1}{1 - 2\tilde{M}_n}\label{mn2}.
  \end{align}
  Because $\tilde{M}_n = o(1)$ by Assumption 5, there exist 
  constants $C_5$ and $C_6$ such that for some sufficiently large $n$, \eqref{mn2} implies
  $|A\tilde{\Omega}B - A\Omega_*B|_1  \leq C_5 \|\tilde{\Omega} - \Omega_* \|_F 
  \xi(p_n, \mathcal{G}) + C_6\tilde{M}_n|[A \Omega_* B]_\mathcal{G}|_1. $
  Combining this with \eqref{main_them2ii},  
  \begin{align*} 
   |\hat{A}_n \tilde{\Omega} \hat{B}_n - A \Omega_*B|_1 & \leq (\tilde{M}_n + 1) 
   \left\{ C_5 \|\tilde{\Omega} - \Omega_* \|_F \xi(p_n, \mathcal{G}) 
   + C_6\tilde{M}_n|[A \Omega_* B]_\mathcal{G}|_1 \right\} + \tilde{M}_n |[A \Omega_* B]_\mathcal{G}|_1.\\
   & =  C_5  (\tilde{M}_n + 1) \|\tilde{\Omega} - \Omega_* \|_F 
   \xi(p_n, \mathcal{G}) + C_6(\tilde{M}^2_n + \tilde{M}_n 
   + \tilde{M}_n C_6^{-1}) |[A \Omega_* B]_\mathcal{G}|_1  
   \end{align*}
  and using that $\tilde{M}_n = o(1)$ with the result from 
  Theorem 2 $(i)$ for $\|\tilde{\Omega} - \Omega_*\|_F$, we obtain the result. 
  \end{proof}

\section{Additional information for the genomic data example}
\begin{figure}[h!]
\centering
\includegraphics[width=14cm]{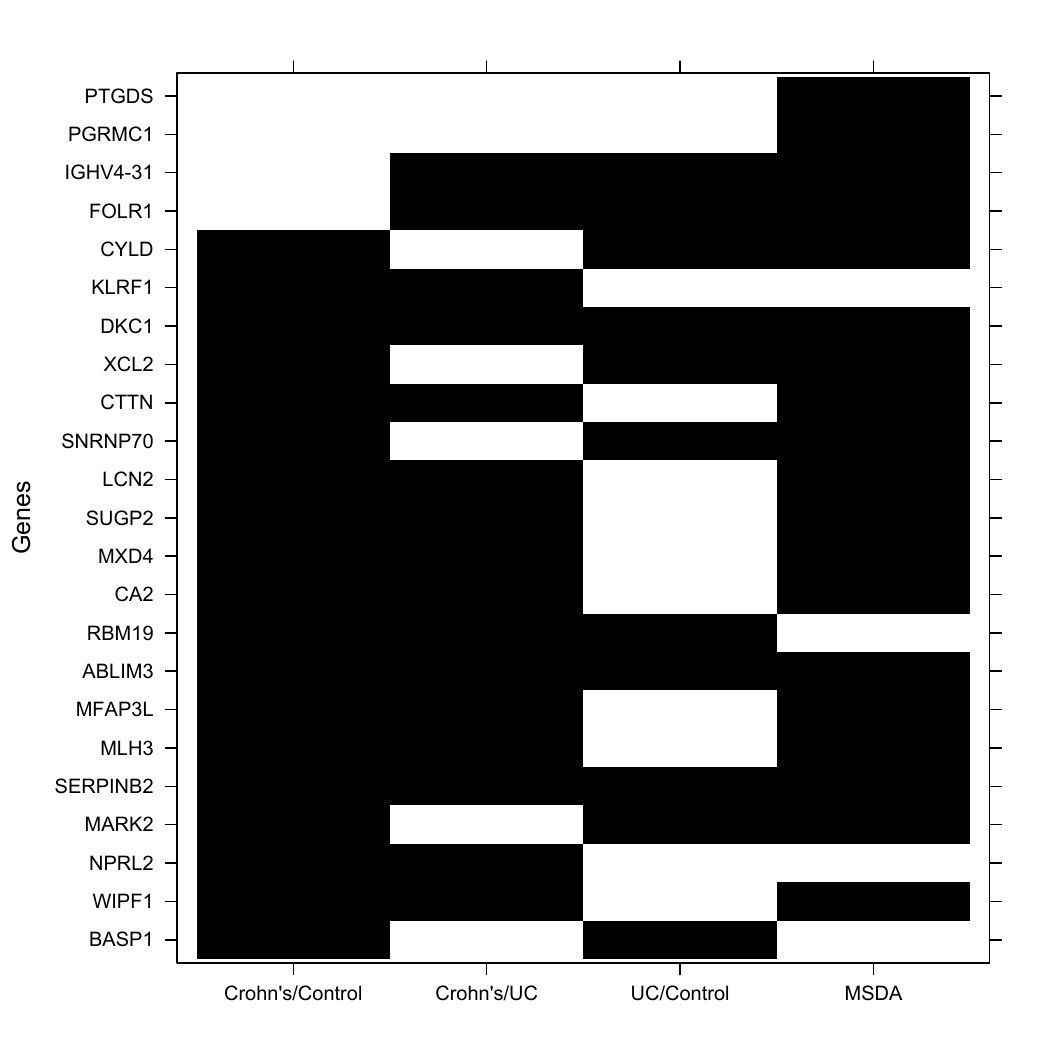}
  \caption{Genes selected as informative for discriminating between the subjects with ulcerative colitis (UC), subjects with Crohn's disease, and controls. Black indicates that the variable was estimated to be informative for discriminating between the response categories. White indicates the variable was estimated to be uninformative. The leftmost three columns are based on the fitted model using our method method, whereas the MSDA column is based on the fitted model using the method of \citet{mai2015multiclass}. }
   \label{fig:variables_selected}
\end{figure}

The data collected by \citet{burczynski2006molecular} were measured on a Affymetrix HG-U133A human GeneChip array from peripheral blood mononuclear cells from 127 patients: 59 with Crohn's disease, 26 with ulcerative colitis, and 42 controls. For our analysis, we used the log-base-2-transformed transcript measurements as predictors. These data are accessible from the Gene Expression Omnibus (https://www.ncbi.nlm.nih.gov/geo/) using accession number GDS1615. The gene names reported in Figure \ref{fig:variables_selected} are those that correspond to the transcripts according to the Gene Expression Omnibus database. 

To compare the genes selected by the method of \citet{mai2015multiclass} to the genes selected by our proposed method, we refit both models to the complete dataset after selecting tuning parameters to minimize the misclassification rate in five-fold cross-validation. In total, twenty-three genes were identified as informative for classification by at least one of the two methods. The method of \citet{mai2015multiclass}, which identifies variables that are informative for discriminating between all pairwise response category comparisons, identified nineteen informative genes. Our method identified twenty-one informative genes, only four of which were identified as informative for all pairwise response category comparisons. 

Only six genes were selected by one method but not the other: \textit{PTGDS} and \textit{PGRMC1} were selected by the method of \citet{mai2015multiclass} but not by our method, whereas 
\textit{KLRF1}, \textit{RBM19}, \textit{NPRL2}, and \textit{BASP1} were selected by our method but the method of \citet{mai2015multiclass}. Both \textit{BASP1} and \textit{RBM19} have been associated with inflammatory bowel disease, which include Crohn's disease and ulcerative colitis, in previous studies. 
For example, \textit{BASP1} was identified as significantly differentially expressed between control patients and patients with inflammatory bowel disease in the study of \citet{mesko2010peripheral}. In \citet{taleban2015ocular}, \textit{RBM19} was significantly associated with ocular extraintestinal manifestations of inflammatory bowel disease at the genome-wide level. 

The majority of genes selected by our method were estimated to be important for discriminating
between only two pairs of the three response categories pairwise comparisons. Thus, difference in specificity may partially explain why our method performed better than the method of \citet{mai2015multiclass} in terms of classification accuracy. For instance, LNC2 is a known biomarker of inflammatory bowel disease (see \citet{toyonaga2016lipocalin} and references therein). Our method estimates that LNC2 is informative for discriminating subjects with Crohn's disease from subjects with ulcerative colitis and controls, but not for discriminating between subjects with ulcerative colitis and controls. The method of \citet{mai2015multiclass} does not allow for this distinction: their method implicitly assumes that informative genes are informative for all three response category pairwise comparisons. 

Our method and the method of \citet{mai2015multiclass} agree on certain genes associated with inflammatory bowel diseases in previous studies, e.g., \textit{SERPINB2}. In the original study from \citet{burczynski2006molecular}, the authors found that \textit{SERPINB2} was the most differentially expressed gene amongst subjects with an inflammatory bowel disease and controls.

\end{document}